\documentclass[11pt]{article}
\usepackage[margin=1in]{geometry}
\usepackage[linesnumbered,lined,boxed,commentsnumbered,noend]{algorithm2e}
\newcommand{\LineIf}[2]{\textbf{if}\ {#1}\ \textbf{then} {#2} }

\usepackage{datetime}
\usepackage{enumitem}
\usepackage[T1]{fontenc}
\usepackage[utf8]{inputenc}
\usepackage{pgfplots}
\usepackage{enumitem}
\usepackage{amsmath, amssymb}
\usepackage{amsthm}
\usepackage{color}
\usepackage{framed}
\usepackage[numbers]{natbib}
\usepackage{cases}
\usepackage{xparse}
\usepackage{xargs}
\usepackage{appendix}
\usepackage{verbatim, xspace}
\usepackage{tikz}
\usetikzlibrary{arrows}
\usetikzlibrary{patterns}
\usetikzlibrary{calc}
\usetikzlibrary{decorations.pathreplacing,angles,quotes}
\usepackage{hyperref}
\usepackage{pgfplots}
\usepackage[framemethod=TikZ]{mdframed}

\usepackage{thmtools}
\usepackage{thm-restate}
\usepackage{cleveref}

\newcommand{\inc}{\ensuremath{\mathtt{inc}}\xspace}
\newcommand{\dec}{\ensuremath{\mathtt{dec}}\xspace}
\newcommand{\nextel}{\ensuremath{\mathtt{next}()}\xspace}
\newcommand{\disj}{\ensuremath{\mathtt{Disj}}\xspace}

\usepackage{xcolor}
\usepackage[colorinlistoftodos,prependcaption,textsize=tiny]{todonotes}
\newcommandx{\unsure}[2][1=]{\todo[linecolor=green,backgroundcolor=green!25,bordercolor=green,#1]{\normalsize #2}}
\newcommandx{\improvement}[2][1=]{\todo[inline,linecolor=blue,backgroundcolor=blue!05,bordercolor=blue,#1]{\normalsize #2}}
\newcommandx{\info}[2][1=]{\todo[linecolor=yellow,backgroundcolor=yellow!25,bordercolor=yellow,#1]{#2}}
\newcommandx{\floatmodel}[2][1=]{\todo[inline,linecolor=red,backgroundcolor=yellow!25,bordercolor=yellow,#1]{#2}}
\newcommandx{\thiswillnotshow}[2][1=]{\todo[disable,#1]{#2}}
\newcommandx{\karol}[2][1=]{\todo[inline,linecolor=green,backgroundcolor=green!25,bordercolor=green,#1]{\normalsize #2}}
\newcommandx{\jesper}[2][1=]{\todo[inline,linecolor=red,backgroundcolor=red!25,bordercolor=red,#1]{\normalsize #2}}

\newtheorem{theorem}{Theorem}
\newtheorem{definition}[theorem]{Definition}
\newtheorem{lemma}[theorem]{Lemma}
\newtheorem{corollary}[theorem]{Corollary}
\newtheorem{claim}[theorem]{Claim}

\newtheorem{observation}[theorem]{Observation}

\newtheorem{remark}[theorem]{Remark}

\newtheorem{assumption}{Assumption}

\numberwithin{theorem}{section}
\numberwithin{lemma}{section}
\numberwithin{claim}{section}
\numberwithin{corollary}{section}
\numberwithin{definition}{section}
\numberwithin{observation}{section}
\numberwithin{remark}{section}

\newcommand{\eps}{\varepsilon}

\newcommand{\Oh}{\mathcal{O}}
\newcommand{\Os}{\Oh^{\star}}
\newcommand{\Oms}{\Omega^{\star}}
\newcommand{\Otilde}{\widetilde{\Oh}}
\newcommand{\Ot}{\Otilde}

\newcommand{\nat}{\mathbb{N}}

\newcommand{\Aa}{\mathcal{A}}
\newcommand{\Bb}{\mathcal{B}}
\newcommand{\Ff}{\mathcal{F}}
\newcommand{\Rr}{\mathcal{R}}
\newcommand{\Ss}{\mathcal{S}}

\newcommand{\Ll}{\mathcal{L}}

\newcommand{\prob}[2]{\mathbb{P}_{#2}\left[ #1 \right]}
\newcommand{\Ex}[1]{\mathbb{E}\left[ #1 \right]}

\newcommand{\ssum}{Subset Sum\xspace}

\newcommand{\defproblem}[3]{
  \vspace{2mm}
  \vspace{1mm}
\noindent\fbox{
  \begin{minipage}{0.95\textwidth}
  #1 \\
  {\bf{Input:}} #2  \\
  {\bf{Task:}} #3
  \end{minipage}
  }
  \vspace{2mm}
}

\newcounter{openquestion}
\newenvironment{openquestion}
{\begin{center}\begin{minipage}{\textwidth}\begin{framed}\refstepcounter{openquestion}\textbf{Question~\theopenquestion:}}	{\end{framed}\end{minipage}\end{center}}

\title{Improving Schroeppel and Shamir's Algorithm for Subset Sum via Orthogonal Vectors}
\date{}

\author{
    Jesper Nederlof\footnote{Utrecht University, The
    Netherlands, \texttt{j.nederlof@uu.nl}. Supported by
    the project CRACKNP that has received funding from the European
    Research Council (ERC) under the European Union’s Horizon 2020 research and
    innovation programme (grant agreement No 853234).}
    \and
    Karol W\k{e}grzycki\footnote{Saarland University and Max Planck Institute for Informatics,
        Saarbr\"ucken, Germany, \texttt{wegrzycki@cs.uni-saarland.de}. 
    This work is part of the project TIPEA that has
    received funding from the European Research Council (ERC) under the European Unions Horizon
    2020 research and innovation programme (grant agreement No. 850979).
Author was also supported Foundation for Polish Science (FNP), by the grants
    2016/21/N/ST6/01468 and 2018/28/T/ST6/00084 of the Polish National Science
    Center and project TOTAL that has received funding from the European
    Research Council (ERC) under the European Union’s Horizon 2020 research and
    innovation programme (grant agreement No 677651).} }

\begin{document}

\maketitle

\begin{picture}(0,0)
\put(-67,-430)
{\hbox{\includegraphics[width=40px]{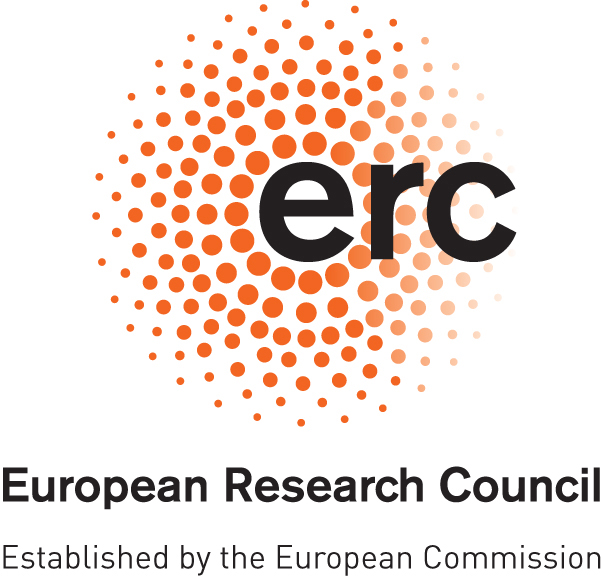}}}
\put(-77,-490)
{\hbox{\includegraphics[width=60px]{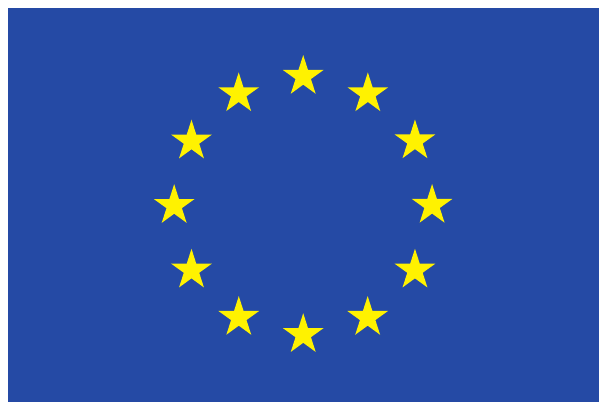}}}
\end{picture}

\thispagestyle{empty}
We present an $\Os(2^{0.5n})$ time and $\Os(2^{0.249999n})$ space randomized
algorithm for solving worst-case Subset Sum instances with $n$ integers. This is the first improvement over the
long-standing $\Os(2^{n/2})$ time and $\Os(2^{n/4})$ space algorithm due to
Schroeppel and Shamir (FOCS 1979).

\medskip 
We breach this gap in two steps: (1) We present a space efficient reduction to
the Orthogonal Vectors Problem (OV), one of the most central problem in
Fine-Grained Complexity. The reduction is established via an intricate
combination of the method of Schroeppel and Shamir, and the representation
technique introduced by Howgrave-Graham and Joux (EUROCRYPT 2010) for designing Subset Sum algorithms for the average case regime. (2) We provide an algorithm for OV that detects an
orthogonal pair among $N$ given vectors in $\{0,1\}^d$ with support size $d/4$
in time $\Ot(N\cdot2^d/\binom{d}{d/4})$. Our algorithm for OV is
based on and refines the representative families framework developed by Fomin, Lokshtanov, Panolan and Saurabh (J. ACM 2016).


\medskip 

Our reduction uncovers a curious tight relation between Subset Sum and OV, because any improvement of our algorithm for OV would imply an improvement over the runtime of Schroeppel and Shamir, which is also a long standing open problem.


\clearpage
\setcounter{page}{1}

\section{Introduction}
The most natural question in computational complexity is: Can an algorithm be improved, or is there some fundamental barrier stopping us from doing so?
A major theme in contemporary research has been to study this question in a
\emph{fine-grained} sense: given an algorithm using $\mathcal{T}$ time (and
$\mathcal{S}$ space) on worst-case instances, can this be improved to
$\mathcal{T}^{1-\eps}$ (or $\mathcal{S}^{1-\eps}$ space), for some $\eps>0$?

Because of the highly challenging nature of finding such improvements, researchers introduced several hypotheses that state that the currently best known algorithms already hit upon a barrier and therefore cannot be improved in the above sense.
Under these hypotheses, many simple algorithms for standard problems in $P$ like $3$-SUM, Edit Distance, or Diameter cannot be significantly improved. 
The fine-grained hardness of the latter two problems is based on the hardness of a problem that is particularly central in the area, called \emph{Orthogonal Vectors}: Given $N$ vectors in $\{0,1\}^d$, detect two orthogonal vectors.
A common hypothesis is that for $d=\omega(\log n)$ the problem cannot be solved
in $\Oh(N^{2-\eps})$ time for some constant $\eps>0$. See for example the survey~\cite{icm-survey}.

For NP-complete problems the situation is slightly different: Although similar fine-grained hypotheses for CNF-SAT and Set Cover have been introduced, they did not prove sufficient yet to rule out improvements of the currently best algorithms for basic NP-complete problems such as Traveling Salesman, Graph Coloring and MAX-3-SAT.
See the survey~\cite{ETH-survey} for some hardness results in this regime.
There may be a good reason for this: While improved polynomial time algorithms
can be naturally used as subroutines for improved exponential time algorithms,
the converse is far less natural. Therefore it is quite plausible that finding
better exponential time algorithms is much easier than finding faster polynomial time algorithms.
And indeed, in the last decade improved algorithms for basic problems such as
Undirected Hamiltonicity~\cite{Bjorklund14} and Graph
Coloring~\cite{BjorklundH06} were found. This motivates the optimism that
for many NP-complete problems the currently best known algorithms can still
be enhanced.

Equipped with this optimism, we study the fine-grained complexity of the following important class of NP-complete problems revolving around numbers.

\paragraph{Subset Sum, Knapsack and Binary Integer Programming}

In the Subset Sum problem, we are given as input a set of integers $\{w_1,\ldots,w_n\}$ and a target $t$.
The task is to decide if there exists a subset $S \subseteq \{1,\ldots,n\}$
such that the total sum of integers $w(S):= \sum_{i \in S} w_i$ is equal to $t$.

In the 1970's, \cite{horowitz-sahni} introduced the \emph{meet-in-the-middle
strategy} and solved Subset Sum in $\Os(2^{n/2})$ time and space.
Since then, it has been a notorious open question to improve their result:
\begin{openquestion}\label{qmain}
	\label{q:mitm}Can Subset Sum be solved in $\Os\left(2^{(1/2-\eps)n}\right)$ time, for $\eps>0$?
\end{openquestion}
A few years later, \cite{schroeppel} gave an algorithm for Subset Sum using
$\Os(2^{n/2})$ time and only $\Os(2^{n/4})$ \emph{space}.
In the last section of their paper, they ask the following:
\begin{openquestion}\label{q:ss}
	Can Subset Sum be solved in $\Os\left(2^{n/2}\right)$ time and $\Os\left(2^{(1/4-\eps)n}\right)$ space, for $\eps>0$?
\end{openquestion}

Both questions seemed to be out of reach until 2010, when
\cite{generic-knapsack1} introduced the \emph{representation technique} and used it
to solve \emph{random instances} of Subset Sum in $\Os(2^{0.337n})$ time.
The main idea behind the representation technique is to artificially expand the
search space such that a single solution has an exponential number $r$ of
representatives in the new search space. This allows us to subsequently restrict
attention to a $1/r$-fraction of the search space, which in some settings can be advantageous.
In the context of Subset Sum, this technique has already inspired improved algorithms for large classes of instances~\cite{stacs2015,stacs2016}, time-space trade-offs~\cite{subsetsum-tradeoff,dinur-tradeoff} and improved polynomial space algorithms~\cite{polyspace-stoc2017}.

Nevertheless, answers to
Questions~\ref{q:mitm} and~\ref{q:ss} for worst-case instances still remained elusive.

\subsection{Our Main Result and Key Insight}

Our main result is a positive answer to the 40-year old open Question~\ref{q:ss}:

\begin{restatable}{theorem}{mainthm}
	\label{main-thm}
	Every instance of Subset Sum can be solved in $\Os(2^{n/2})$ time and
	$\Os(2^{0.249999n})$ space by a randomized Monte Carlo
	algorithm with constant success probability.
\end{restatable}

The result implies an analogous
space improvement for Knapsack and Binary Integer Programming (see
Corollary~\ref{cor:knbip}).\footnote{\label{fn:apppd}See Appendix~\ref{sec:problems} for definitions of all problems considered in this paper.}
To explain our key ideas and their combination with existing methods, the following problem is instrumental:

\newcommand{\wOV}{\ensuremath{\textnormal{WOV}}}
\defproblem{\underline{Weighted Orthogonal Vectors} \hfill  \textcolor{gray}{(shorthand notation: $\wOV(N,d,h)$)}}
{Families of $N$ weighted sets $\mathcal{A}, \mathcal{B} \subseteq \binom{[d]}{h} \times \mathbb{N}$ of Hamming weight $h$, target integer $t$}
{Detect $(A,w_A) \in \mathcal{A}$ and $(B,w_B) \in \mathcal{B}$ such that $A$ and $B$ are disjoint and $w_A+w_B=t$}

The starting point is the $\Os(2^{n/2})$ time and space algorithm by~\cite{horowitz-sahni}.
Their algorithm can be seen as a reduction to an instance of $\wOV(2^{n/2},0,0)$. Since $d=0$, this is an instance of 2-SUM,\textsuperscript{\ref{fn:apppd}} and the runtime follows by a linear time algorithm for 2-SUM.

In contrast, the representation technique by \cite{generic-knapsack1} can also
be thought of as a reduction from instances of Subset Sum to WOV, but with the
assumption that the Subset Sum instance does not have \textit{additive structure}.\footnote{Specifically, this means that $|w(2^{[n]})| \geq 2^{(1-\varepsilon)n}$ for some small $\varepsilon >0$, where $w(2^{[n]})$ denotes $\{w(X) : X \subseteq [n] \}$. }
In a follow-up work, \cite{stacs2016} loosen the assumption of
\cite{generic-knapsack1} and show that their reduction applies whenever there is a small subset of $d$ weights without additive structure. Their work implies a reduction from every\footnote{Actually not every instance, but instances where the reduction fails can be solved quickly by other means.} instance of Subset Sum to $\wOV(N,d ,d/ 4)$, where $N=2^{n/2}\binom{d}{d / 4}/2^d$ and $d/n > 0$ is a small (but fixed) positive constant.

Note that the two above reductions feature an intriguing trade-off between the \emph{size} $N$ and the \emph{dimension} $d$ of the produced instance, and the natural question is how the worst-case complexities of solving these instances as quick as possible compare.
Our first step towards proving Theorem~\ref{main-thm} is to show that this trade-off is \textbf{tight}, unless Question~\ref{q:mitm} is answered positively:
\mdfsetup{%
    nobreak=true,
	middlelinecolor=gray,
	middlelinewidth=1pt,
	backgroundcolor=gray!10,
    innertopmargin=7pt,
	roundcorner=5pt}
\begin{mdframed}
	\textbf{Key Insight:} There is an algorithm for $\wOV$ whose run time dependency in $d$ matches the instance decrease in $d$ in the reduction from~\cite{stacs2016}.
	In particular, $\wOV(N,d,d/4)$ can be solved in $\Ot(N\cdot
    2^d/\binom{d}{d/4})$ time and $\Ot(N+2^d)$ space (see
    Theorem~\ref{ov-lemma}).
\end{mdframed}

This insight has two interesting immediate consequences. First, it provides an
avenue towards resolving Question~\ref{q:mitm}, because a positive answer to
this question is implied by an improvement of our algorithm even for the unweighted version of $\wOV(N,d,d/4)$. To the best of our knowledge, such an improvement is entirely consistent with all the known hypotheses on (low/moderate/sparse) versions of the Orthogonal Vectors problem~\cite{ChenW19, GaoIKW19}. In fact, to answer Question~\ref{q:mitm} affirmatively we only need an improvement for the regime $2^d/\binom{d}{d/4}\leq N \leq 2^d$, while previous hypotheses address the regime where $d / \log N$ tends to infinity.

Second, a combination of the reduction from~\cite{stacs2016} and our
algorithm for $\wOV(N,d,d/4)$ would give an algorithm for Subset Sum that runs
in $\Os(2^{n/2})$ time and $\Os(2^{(1/2-\delta)n})$ space, for some small
$\delta >0$. While this is not even close to the memory improvement of~\cite{schroeppel}, one may hope that by adding the ideas~\cite{schroeppel} on top of this approach results in a
better memory usage. Notwithstanding the significant hurdles that need to be overcome to make these two methods combine, this is exactly how we get the improvement in Theorem~\ref{main-thm}.

\subsection{The Representation Technique Meets Schroeppel and Shamir's Technique}
We now give a high level proof idea of Theorem~\ref{main-thm}.
While our conceptual contribution lies in the aforementioned key insight, our
main \emph{technical} effort lies in showing that indeed the representation
technique and the algorithm of~\cite{schroeppel} can be combined to get a space efficient reduction from Subset Sum to $\wOV$. 

Following~\cite{horowitz-sahni}, the method from~\cite{schroeppel} can also be
seen as a reduction from Subset Sum to an instance of $\wOV(2^{n/2},0,0)$, but it
is an \emph{implicit} one: The relevant vectors of the instance can be
enumerated quickly by decomposing the search space of $2^{n/2}$ vectors into a
Cartesian product of two sets of $2^{n/4}$ vectors, and generating all vectors in a
useful order via priority queues. See Section~\ref{sec:prel} for a further explanation.
Thus, to prove Theorem~\ref{main-thm}, we aim to generate the relevant parts of
the instance $\mathcal{I}$ of $\wOV(N,d,d/4)$ defined by the representation
technique efficiently, using priority queues of size at most $\Os(2^{0.249999n})$.

Unfortunately, the vectors from the instance $\mathcal{I}$ defined by the representation technique are elements of a search space of size $2^{(1/2+\Omega(1))n}$; its crux is that there are only $2^{(1/2-\Omega(1))n}$ vectors in the instance because we only have vectors with a fixed inner product with the weight vector $(w_1,\ldots,w_n)$.\footnote{Some knowledge of the representation technique is required to understand this in detail; We explain the representation technique in Section~\ref{sec:ss-via-ov}.} Thus a straightforward decomposition of this space into a Cartesian product will give priority queues of size $2^{(1/4+\Omega(1))n}$.

To circumvent this issue, we show that we can apply the representation technique
\emph{again} to generate the vectors of the instance $\mathcal{I}$ efficiently
using priority queues of size $\Os(2^{0.249999n})$. While the representation technique was already used in a multi-level fashion in several earlier works~ (see e.g.~\cite{subsetsum-tradeoff, generic-knapsack1}), an important ingredient of our algorithm is that we apply the technique in different ways at the different levels depending on the structure of the instance.

\subsection{Additional Results and Techniques}
Our route towards Theorem~\ref{main-thm} as outlined above has the following by-products that may be considered interesting on their own. The first one was already referred to in the `key insight':

\paragraph{An algorithm for Orthogonal Vectors}
\newcommand{\OV}{\ensuremath{\textnormal{OV}}}

A key subroutine in this paper is the following algorithm for Orthogonal
Vectors. We let $\OV(N,d,h)$ refer to the problem $\wOV(N,d,h)$ restricted to
unweighted instances (that is all involved integers are zero).

\begin{restatable}{theorem}{ovmainthm}
	\label{ov-simplified}
	There is a Monte-Carlo algorithm solving $\OV(N,d,d/4)$ using 
	$\Ot\left( N \cdot 2^{d}/\binom{d}{d/4} \right)$ time and $\Ot(N+2^d)$ space.
\end{restatable}

An easy reduction shows the same runtime and space usage can be obtained for $\wOV(N,d,d/4)$.
Our algorithm for Orthogonal Vectors uses the general blueprint of an algorithm
by~\cite{fomin-jacm2016} (which in turn builds upon ideas
from~\cite{bollobas1965generalized, Monien83}). However, to ensure that the
algorithm for Orthogonal Vectors combined with our methods result in an $\Os(2^{n/2})$
time algorithm for Subset Sum, we need to refine their method and analysis. 

To facilitate our presentation, we consider a new communication complexity-related parameter of \emph{$1$-covers of a matrix}
that we call the `\emph{sparsity}'. We show that a $1$-cover of low sparsity of a specific \emph{Disjointness Matrix}
implies an efficient algorithm for Orthogonal Vectors, and we exhibit a $1$-cover of low sparsity of the disjointness matrix. We also prove that our
$1$-cover has nearly optimal sparsity. This means that Question~\ref{main-thm} cannot be resolved directly via our route combined with improved $1$-covers.
Additionally, we use several preprocessing techniques to ensure that the
space usage of our algorithm is only $\Ot(N + 2^d)$,
which is crucial to Theorem~\ref{main-thm}. 

\paragraph{Reduction to weighted $P_4$}
While we do not resolve Question~\ref{q:mitm}, our approach
provides new avenues by reducing it to typical questions in the study of
fine-grained complexity of problems in the complexity class P. 
Our new reductions enable us to show a new
connection between Subset Sum and the following graph problem:
In the Exact Node Weighted $P_4$ problem one is
given an undirected graph $G=(V,E)$ with vertex weights and the task is to
decide whether there exists a path on four vertices with weights summing to $0$ (see
also Appendix~\ref{sec:problems}).

\begin{restatable}{theorem}{weightedpfour}
	\label{thm:weightedp4}
    If Exact Node Weighted $P_4$ on a graph $G=(V,E)$ can be solved in
    $\Oh(|V|^{2.05})$ time, then Subset Sum can be solved in $\Os(2^{(0.5 - \delta)n})$ randomized time for some $\delta > 0$.
\end{restatable}

In comparison to the straightforward reduction from Subset Sum to $4$-SUM, our
reduction creates a set of integers with an additional path constraint.
Thus a possible attack towards resolving Question~\ref{q:mitm} is to design a
quadratic time algorithm for Exact Node Weighted $P_4$ (or more particularly, only
for the instances of the problem generated by our reduction).

The na\"ive algorithm for Exact Node Weighted $P_4$ works in $\Ot(|V|^3)$ time. To the
best of our knowledge the best algorithm for this problem runs in $\Ot(|V|^{2.5})$
when $\omega = 2$~\cite{karl} (where $\omega$ is the exponent of currently the fastest algorithm for matrix multiplication).\footnote{Briefly described, reduce a problem instance on $(V,E)$ to the problem of finding a triangle in an unweighted graph on $|V|^2$ edges: One vertex of the triangle represents the two extreme vertices of the path and the sum of the weights of the two first vertices on the path, and the other two vertices of the triangle represent the inner vertex. This instance of unweighted triangle can be solved in $\Ot(|V|^{2.5})$ with standard methods, assuming $\omega=2$.} 
On the lower bounds side, using the `vertex minor' method from~\cite{AbboudL13}
it can be shown that the problem of detecting triangles in a graph can be
reduced to Exact Node Weighted $P_4$~\cite{amir}. This explains that obtaining a
quadratic time algorithm may be hard (since it is even hard to obtain for detecting
triangles). 
However, detecting triangles in a graph is known to be solvable in $\Ot(|V|^{\omega})$ time,  Therefore it is still
justified to aim for a $\Ot(|V|^\omega)$ time algorithm for Exact Node Weighted
$P_4$. We leave it as an intriguing open question whether Exact Node Weighted $P_4$ can be solved in $\Ot(|V|^\omega)$
time.

\paragraph{More general problems} 

Known reductions from~\cite{nederlof-mfcs2012} combined with Theorem~\ref{main-thm} also imply the following improved algorithms for generalizations of the Subset Sum problem (see Appendix~\ref{sec:problems} for their definitions):

\begin{corollary}\label{cor:knbip}
	Any instance of Knapsack on $n$ items can be solved in $\Os(2^{n/2})$ time
    and $\Os(2^{0.249999n})$ space, and any instance of Binary Integer
    Programming with $n$ variables and $d$ constraints with maximum absolute
    integer value $m$ can be solved in $\Os(2^{n/2}(\log(mn)n)^d)$ time and
    $\Os(2^{0.249999n})$ space.
\end{corollary}

\subsection{Related Work}

It was shown in~\cite{schroeppel} that Subset Sum admits a
time-space tradeoff, i.e. an algorithm using $\mathcal{S}$ space and $2^{n}/\mathcal{S}^2$ time for any $\mathcal{S} \le
\Os(2^{n/4})$. This tradeoff was improved by
\cite{subsetsum-tradeoff} for almost all tradeoff parameters (see also
\cite{dinur-tradeoff}).
We mention in the passing that as direct consequence of Theorem~\ref{main-thm}, 
the Subset Sum admits a time-space tradeoff using
$2^{n}/\mathcal{S}^{0.5/0.249999}$ time and $\mathcal{S}$ space, for any
$\mathcal{S} \leq \Os(2^{0.249999n})$, but the obtained parameters are only better than the previous works for $\mathcal{S}$ chosen closely to its maximum. See Appendix~\ref{proof-of-lem:sumsetenum} for a proof.

In~\cite{stacs2015}, the authors considered \ssum parametrized by the parameter
$\beta$ (which is defined as the largest number of subsets of the input integers
that yield the same sum) and obtained an algorithm running in time $\Os(2^{0.3399n}\beta^4)$. Subsequently,
\cite{stacs2016} showed that one can get a faster algorithm for \ssum than
meet-in-the-middle if $\beta \le 2^{(0.5 - \eps)n}$ or $\beta \ge 2^{0.661n}$.
Recently, \cite{polyspace-stoc2017} gave an algorithm for \ssum running in
$\Os(2^{0.86n})$ time and polynomial space, assuming random access to a random oracle.


From the pseudopolynomial algorithms perspective, Subset Sum has also been
subject of recent stimulative
research~\cite{subset-sum-lower2,bringmann-soda,subset-sum-lower,sosa-wu,
koiliaris-soda,simplified-koiliaris,pseudopolynomial-polyspace}. These
algorithms use plethora of ingenious algorithmic techniques, e.g., dynamic
programming, color-coding and the Fast Fourier Transform.

\paragraph{Average case complexity and representation technique}

In a breakthrough paper, \cite{generic-knapsack1} introduce the
\emph{representation technique} and showed $\Os(2^{0.337n})$ time and
$\Os(2^{0.256n})$ space algorithm for an
\ssum in average-case setting. It was improved by \cite{generic-knapsack2} who
gave an algorithm running in $\Os(2^{0.291n})$ time and space.

The representation technique already found several applications in the worst-case setting for other problems (see \cite{esa19,nederlof-setcover,bin-packing}).


\paragraph{Orthogonal Vectors}

Naively, the Orthogonal Vectors problem can be solved in $\Oh(d N^2)$ time. For large $d$, only a slightly faster algorithm that runs
in time $N^{2-1/\Oh(\log(d/\log{N}))}$ time is
known~\cite{abboud-williams-yu-15,chan-williams-16}. The assumption that for $d
= \omega(\log{N})$ there is no $\Oh(N^{2-\eps})$ time algorithm any $\eps > 0$
is a central conjecture of fine-grained complexity (see~\cite{icm-survey} for an
overview).

In this paper, we are mainly interested in linear (in $N$) time algorithms for OV.
It was shown in~\cite{ov-seth-hard} that OV cannot
be solved in $\Oh(N^{2-\eps} \cdot 2^{o(d)})$ time for any $\eps > 0$ assuming SETH. In~\cite{bjorklund-ov} an algorithm for OV was given that runs in $\Ot(D)$ time, where $D$ is the total number of vectors whose support is a subset of the support of an input vector.


\subsection{Organization}

In this paper we heavily build upon previous literature, and in particular the
representation technique as developed in~\cite{stacs2016, generic-knapsack2}.
Therefore, we introduce the reader to this technique in
Section~\ref{sec:ss-via-ov}. At the end of Section~\ref{sec:ss-via-ov}, we also use the introduced terminology of the representation technique to explain the new steps towards proving Theorem~\ref{main-thm}.

The remainder of the paper is devoted to formally support all claims made. Necessary preliminaries are provided in Section~\ref{sec:prel}; in Section~\ref{main-proof} we present the
proof of Theorem~\ref{main-thm}, and the reduction to Exact Node Weighted $P_4$
from Theorem~\ref{thm:weightedp4} is given in Section~\ref{sec:p4}.
Section~\ref{sec:ov} contains the proof of (a generalization of)
Theorem~\ref{ov-simplified}.  In Appendices~\ref{proof-of-lem:sumsetenum}
to~\ref{sec:problems} we provide, respectively, various short omitted proofs,
an inequality relevant for the runtime of our algorithms and a list of problem statements.

\section{Introduction to the Representation Technique, and its Extensions}
\label{sec:ss-via-ov}

This section is devoted to explain the 
representation technique (and its extensions from \cite{stacs2016}) and
will serve towards a warm up towards the formal proof of Theorem~\ref{main-thm}.

\subsection{The Representation Technique with a Simplified Assumption}
We fix an instance $w_1,\ldots,w_n,t$ of Subset Sum. A \emph{perfect mixer} is a subset $M \subseteq [n]$ such that for every distinct subsets $A_1,A_2 \subseteq M$ we have $w(A_1) \neq w(A_2)$.\footnote{This notion will be generalized to the notion of an $\eps$-mixer in Definition~\ref{def:epsmix}.} To simplify the explanation in this introductory section, we will make the following assumption about the Subset Sum instance:

\begin{assumption}\label{assumption}
	If $w_1,\ldots,w_n,t$ is a YES-instance of Subset Sum, then there is a perfect mixer $M \subseteq [n]$ and a set $S$ with $w(S)=t$ such that $|M \cap S|=|M|/2$.
\end{assumption}

A mild variant of Assumption~\ref{assumption} can be made without loss of generality since relatively standard extensions of the method by~\cite{schroeppel} can be used to solve the instance more efficiently if it does not hold. We discuss the justification of Assumption~\ref{assumption} more later.
We now illustrate the representation technique by outlining proof of the following statement:
\begin{theorem}\label{simple-secondmain-thm}
An instance of Subset Sum satisfying Assumption~\ref{assumption} can be reduced
to an equivalent instance of $\wOV\left(2^{n/2}\binom{|M|}{|M|/4}/2^{|M|},|M|, |M|/4\right)$ in the linear (in the
size of the output) randomized time.
\end{theorem}

\begin{algorithm}
	\DontPrintSemicolon
	\SetKwInOut{Input}{Algorithm}
	\SetKwInOut{Output}{Output}
	\Input{$\mathtt{RepTechnique}$($w_1,\ldots,w_n,t,M$)}
	\Output{Instance of weighted orthogonal vectors }
    Arbitrarily partition $[n] \setminus M$ into $L$ and $R$ such that $|L|=|R|=(n-|M|)/2$  \\
	Pick a random prime $p$ of order $2^{|M|/2}$\\
	Pick a random $x \in \mathbb{Z}_p$ \\
	Construct the following sets:\label{line:construct}
    \vspace*{-0.2cm}
	\begin{align*}
	\Ll := &\Big\{(A_1 \cap M,w(A_1)) &\Big|&  &A_1 \in 2^{L \cup M} &:& |A_1 \cap M| = |M|/4\ \  &\text{ and}& w(A_1) \equiv_p x \Big\}&\\
	\Rr := &\Big\{ (A_2 \cap M,w(A_2)) &\Big|& &A_2 \in 2^{R \cup M} &:&  |A_2 \cap M| = |M|/4\ \  &\text{ and }& w(A_2) \equiv_p t - x \Big\} &
	\end{align*}\\
    \vspace*{-0.2cm}
	\Return the instance $(\Ll,\Rr, t)$ of weighted orthogonal vectors
	\caption{Pseudocode of Theorem~\ref{simple-secondmain-thm}}
	\label{secondmain-alg}
\end{algorithm}

The reduction from Theorem~\ref{simple-secondmain-thm} is described in Algorithm~\ref{secondmain-alg}, and uses the standard notation $\equiv_p$ to denote equivalence modulo $p$.
We now describe the intuition of the algorithm.  For partition of $[n]$ into
$L,M,R$, it \emph{expands} the search space by looking for pairs $(A_1,A_2)$ where
$A_1 \in 2^{L \cup M}$, $A_2 \in 2^{R \cup M}$ and both $A_1$ and $A_2$ use $|M|/4$ elements of $M$. This is
useful since the assumed solution $S$ is \emph{represented} by the $\binom{|M
\cap S|}{|M|/4}\approx 2^{|M|/2}$ partitions $(A_1,A_2)$ of $S$ that are in the
expanded search space. Together with Assumption~\ref{assumption}, this allows us in turn to narrow down the search space by restricting the search to look only for pairs $(A_1,A_2)$ satisfying $w(A_1)\equiv_p x$, and thus $w(A_2)\equiv_p t- x$. Thus, the algorithm enumerates all candidates for $A_1$ and $A_2$ in respectively $\Ll$ and $\Rr$ and the instance of weighted orthogonal vectors detects a disjoint pair of candidates with weights summing to $t$.

One direction of the correctness of the algorithm follows directly: If the produced instance of weighted orthogonal vectors is a YES-instance, the union of the two found sets is a solution to the Subset Sum instance. 

Conversely, we claim that if the instance of Subset Sum is a YES-instance and Assumption~\ref{assumption} holds, then with good probability the output instance of Weighted Orthogonal Vectors is a YES-instance. Let $S$ be the solution of the Subset Sum instance, so $w(S)=t$ and $|S \cap M|=|M/2|$ by Assumption~\ref{assumption}. Note that 
\[
    W := \big|\big\{ w(\tilde{A}_1 \cup (L \cap S)) : \tilde{A}_1 \in \binom{M \cap S}{|M \cap S|/2}  \big\}\big| = \binom{|M|/2}{|M|/4},
\]
because there are $\binom{|M|/2}{|M|/4}$ possibilities for $\tilde{A}_1$ and
$w(\tilde{A}_1)$ is different for each different $\tilde{A}_1$ by the perfect mixer property of $M$.
Therefore, there are $\binom{|M|/2}{|M/4}$ possibilities for $A_1 := A'_1 \cup (L
\cap S)$ and each $w(A_1)$ is different.

By standard properties of hashing modulo random prime numbers (see e.g.
Lemma~\ref{lem:epsmixer} for a general statement), we have that the expected
size of $\{ y \mod p : y \in W \}$ is also approximately of cardinality
$\binom{|M|/2}{|M|/4}$.\footnote{This uses that all numbers are
single-exponential in $n$, but this can be assured with a standard hashing
argument.} Therefore the probability that $x$ is chosen such that $x \equiv_ p
y$ for some $y \in W$ is $\binom{|M|/2}{|M|/4}/2^{|M|/2}\geq
\Omega(\frac{1}{|M|})$. If we let $A_1 \in \Ll$ be the set with $w(A_1)=y$ then
since $w(S\setminus A_1)\equiv_p t - y$, $S \setminus A_1 \in \Rr$ and the pair
$(A_1,S\setminus A_2)$ is a solution to the weighted orthogonal vectors problem. In general this happens with probability at least $1/n$.

\begin{figure}[ht!]
	\centering
	\includegraphics[width=0.7\textwidth]{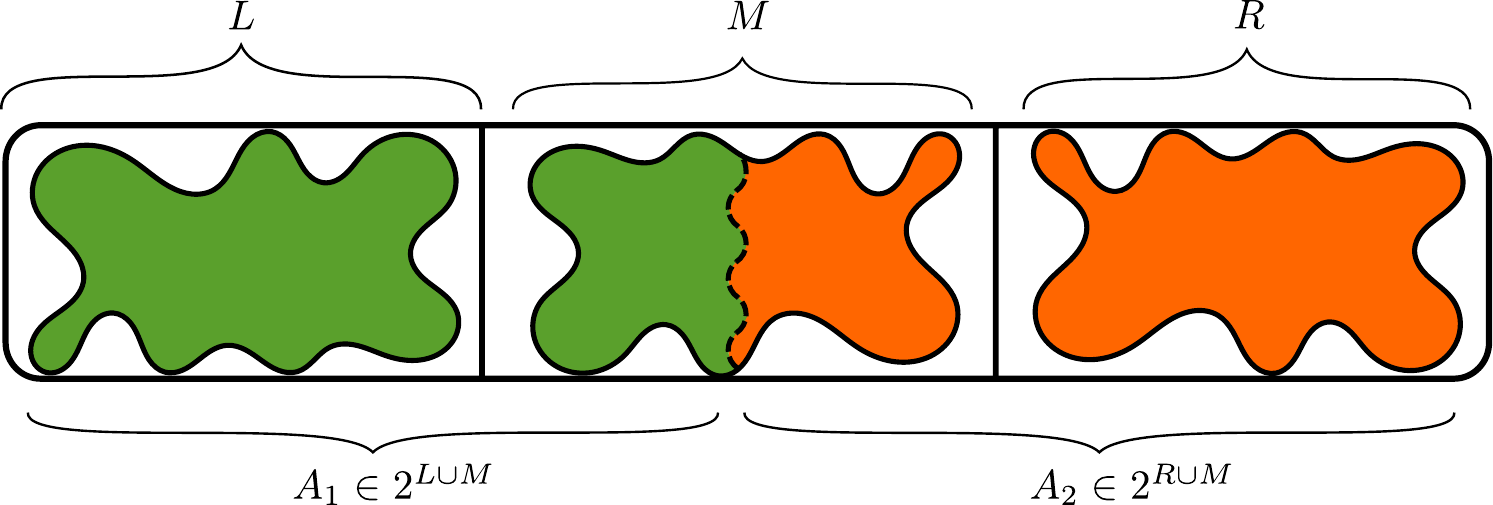}
	\caption{The green and orange regions represent a
		solution, i.e., a set $S$ such that $w(S) = t$.
		There are $\binom{|M \cap S|}{|M \cap S|/2}$ pairs $A_1 \in 2^{L \cup M}$ and $A_2 \in 2^{R \cup M}$, such that
		$A_1 \cup A_2 = S$ and $A_1 \cap A_2 = \emptyset$. Because $M$ is a
    perfect mixer, the number of possible values that $w(A_1 \cap M)$ can take is also $\binom{|M\cap S|}{|M \cap S|/2}$.}
	\label{fig:proof-obs}
\end{figure}

Now we discuss the runtime and output size.
At Line~\ref{line:construct} we construct $\Ll$ and $\Rr$.
Since the number of possibilities of $A_1$ is $2^{|L|}\binom{|M|}{|M|/4}$ and
each such $A_1$ satisfies $w(A_1)\equiv_p x$ with probability $p$ (taken over
the random choices of $x$), we have that the expected sizes of $\Ll$ (and
similarly, of $\Rr$) is $2^{n/2}\binom{|M|}{|M|/4}/2^{|M|}$, as claimed. By
standard pseudo-polynomial dynamic programming techniques (see
e.g.~\cite{stacs2016}), Line~\ref{line:construct} can be performed in $\Oh(p+|\Ll|+|\Rr|)$ time, and thus the claimed run time follows.

\subsection{Representation Technique with Non-Simplified Assumption} \label{subsec:assump}
Assumption~\ref{assumption} is oversimplifying our actual assumptions, and actually only a weaker assumption is needed to apply the representation method. We call any set $M$ that satisfies $|w(2^{M})| = 2^{(1-\eps)|M|}$ an \emph{$\eps$-mixer} (see also Definition~\ref{def:epsmix}).
Denoting $w(\Ff)$ for $\{w(X) : X \in \Ff \}$, the assumption is
\begin{assumption}\label{assumptionreal}
	If $w_1,\ldots,w_n,t$ is a YES-instance of Subset Sum, then there is an $\eps$-mixer $M$, and a set $S$ with $w(S)=t$ such that
$|(\tfrac{1}{2}-\eps' )|M| \leq |M\cap S| \leq (\tfrac{1}{2}+\eps' )|M|$, for some small positive  $\eps,\eps'$.
\end{assumption}
To note that the representation technique introduced above still works with
these relaxed assumptions, we remark that it can be shown that if, $M$ is an
$\eps$-mixer, then $w(\binom{M \cap S}{i}) \geq 2^{(1-f(\eps,\eps'))|M \cap S|}$ for some $\eps,\eps'$ and unknown $i$. Thus in the representation technique we can split the solution in sets $A_1,A_2$ where $|A_1 \cap M|=i$ and $|A_2 \cap M|= |S \cap M|-i$ and use a prime $p$ of order $2^{(1-f(\eps,\eps'))|M \cap S|}$ for some function $f$ that tends to $0$ when $\eps$ and $\eps'$ tend to $0$.

The advantage of the relaxed assumptions in Assumption~\ref{assumptionreal} is that the methods from~\cite{horowitz-sahni,schroeppel} can be extended such that it solves any instance that does \emph{not} satisfy the assumptions exponentially better in terms of time and space than in the worst-case. This allows us to make these assumptions without loss of generality when aiming for general exponential improvements in the run time (or space bound).

For example, in the approach by~\cite{schroeppel}, in some steps of the
algorithm we only need to enumerate subsets of cardinality bounded away from
half of the underlying universe; or in some other steps of the algorithm we can
maintain smaller lists with sums generated by subsets. While these extensions
are not entirely direct, we skip a detailed explanation of them in this
introductory section (see Section~\ref{sec:prepro} for formal statements).

\subsection{Our Extensions of the Representation Technique Towards Theorem~\ref{main-thm}}\label{subsec:intui}

Having described the representation technique, we now explain our approach in more detail.

\paragraph{Setting up the representation technique to reduce the space usage.} 
Now, we present the intuition behind the space reduction of Theorem~\ref{main-thm}. In the previous
subsection we constructed an instance $\Ll,\Rr$ of weighted orthogonal vectors of expected size
$\Os(2^{n/2-\Omega(|M|)})$ such that with good probability there exist $A_1 \in \Ll$ and
$A_2 \in \Rr$ with $w(A_1 \cup A_2) = t$ and $A_1 \cap A_2 = \emptyset$ (if
the answer to Subset Sum is \emph{yes}). We combine this approach with the
approach from~\cite{schroeppel} and aim to efficiently enumerate this instance
of weighted orthogonal vectors instance. To do so, we apply the representation
method two times more, and are able to construct $4$ sets $\Ll_1,\Ll_2,\Rr_1,\Rr_2 \subseteq 2^{[n]}$ with the following
properties:

\begin{enumerate}[align=left, font=\normalfont, label=(\roman*)]
    \item With good probability there exist pairwise disjoint $A_1 \in \Ll_1, A_2 \in \Ll_2,
        A_3 \in \Rr_2, A_4 \in \Rr_1$, such that $w(A_1 \cup A_2 \cup A_3 \cup A_4) = t$, if the Subset Sum instance is a YES-instance,
    \item The expected size of each $\Ll_1,\Ll_2,\Rr_1,\Rr_2$ is $\Os(2^{n/4-\Omega(|M|)})$.
\end{enumerate}

The sets will have the property that elements in $\Ll$ are formed by pairs from
$\Ll_1 \times \Ll_2$ and elements in $\Rr$ are formed by pairs from $\Rr_1
\times \Rr_2$. But in contrast to the technique from~\cite{schroeppel},  the lists $\Ll$ and $\Rr$ can not be easily decomposed into a Cartesian product of two sets of size $\sqrt{|\Ll|}$ and $\sqrt{|\Rr|}$. To overcome this issue, we apply the representation method \emph{again} to enumerate the elements of $\Ll$ and $\Rr$ quickly.
\begin{figure}
	\center
	\includegraphics[width=0.75\textwidth]{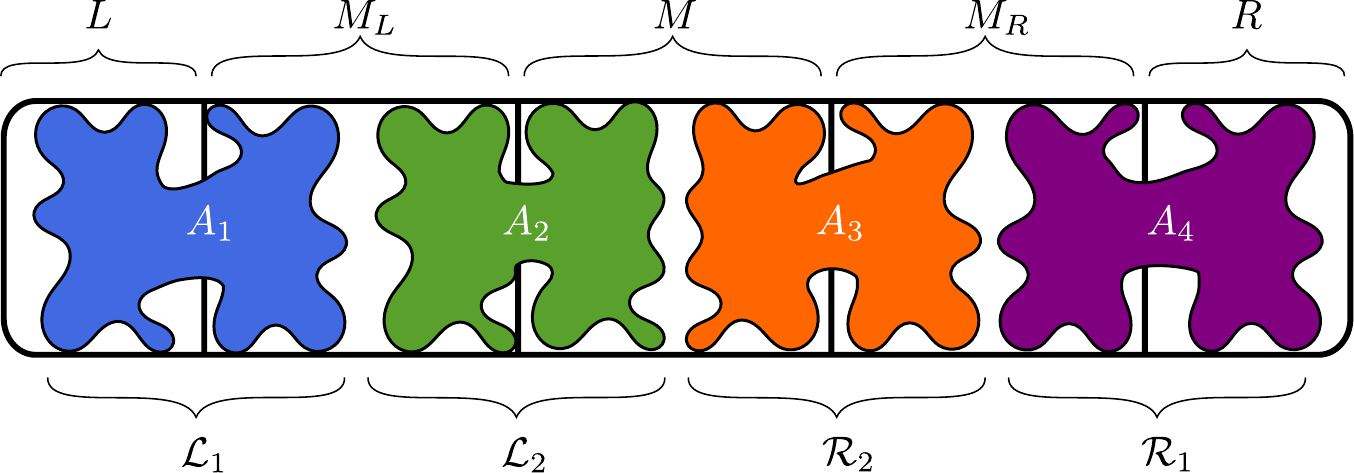}
	\caption{\label{fig:decomposition} The decomposition of the instance into
        $[n] = L \uplus M_L \uplus M \uplus M_R \uplus R$ and decomposition of the solution $S=A_1
    \uplus A_2 \uplus A_3 \uplus A_4$.}
	\label{fig:part}
\end{figure}
In particular, to construct the sets $\Ll_1,\Ll_2,\Rr_1,\Rr_2$, we partition the instance into $L, M_L , M , M_R , R$. See also Figure~\ref{fig:part}.  Here $M$ is assumed to be an 
$\eps$-mixer that is used for the representation technique on `first level': at this level we check whether a pair $\Ll \times \Rr$ forms a solution.
The set $M_L$ is assumed to be an $\eps_L$-mixer and the set $M_R$ is assumed to be an $\eps_R$-mixer, and these sets are used for two applications of the representation technique on the `second level': At this level we check whether a pair in $\Ll_1 \times \Ll_2$ forms an element of $\Ll$ (and similarly, whether a pair in $\Rr_1 \times \Rr_2$ forms an element of $\Rr$). If any of the assumptions fail, relatively direct extensions of the methods from~\cite{schroeppel} can again solve the instance more efficiently in a way similar to how we justified Assumption~\ref{assumptionreal}, so these assumptions are without loss of generality.

\paragraph{Maintaining the $\Os(2^{n/2})$ time bound}

After we construct sets $\Ll_1,\Ll_2,\Rr_2,\Rr_1$ as claimed in property (i) we
combine the approach of~\cite{schroeppel} with our Orthogonal Vectors algorithm
to obtain the $\Os(2^{n/2})$ running time.
In particular, we store the elements of $\Ll_1,\Ll_2,\Rr_1,\Rr_2$ in priority queues ordered by the weight in order to enumerate the elements of $\Ll$ and $\Rr$ in the correct order.
By our applications of the representation technique, we only need to assure
disjointness between pairs between $\Ll_1 \times \Ll_2$, $\Rr_1 \times \Rr_2$
and $\Ll_2 \times \Rr_2$ and therefore the time to check disjointness is the same as in the
normal application of the representation technique as described at the beginning of this section.

Unfortunately, the relaxed assumptions from Assumption~\ref{assumptionreal} cause issues here because we need to consider unbalanced partitions of $M \cap S$, $M_L \cap S$, $M_R \cap S$ and the constants $\eps,\eps_L,\eps_M$ give rise to different primes in our application of the representation technique.
Without additional care, the overhead in the runtime implied by these issues would lead to an undesired time bound of $\Os(2^{(0.5+\eps)n})$ for arbitrarily small constant $\eps >0$.

To address these complications, we analyse our algorithm in such a way that if
$|w(2^{M_L})|$ or $|w(2^{M_R})|$ is significantly smaller than $|w(2^M)|$, then we get an improved runtime. Note this can be assumed by switching the roles of $M_L,M_R,M$. Additionally, we provide a general runtime for solving Weighted Orthogonal Vectors instances with vectors with general support size.

\section{Preliminaries}
\label{sec:prel}

Throughout the paper we use the $\Os(\cdot)$ notation to hide factors polynomial in the
input size and the $\Ot(\cdot)$ notation to hide polylogarithmic factors in the input
size; which input this refers to will always be clear from the context. We also use $[n]$ to denote the set $\{1,\ldots,n\}$.  We use the binomial
coefficient notation for sets, i.e., for a set $S$ the symbol ${S \choose k}$ denotes the set of all subsets of the set $S$ of
size exactly $k$. For a modulus $m \in \mathbb{Z}_{\ge 1}$ and $x,y \in
\mathbb{Z}$ we write $x \equiv_m y$ to
indicate that $m$ divides $x-y$.
If $X \subseteq [n]$, we denote $w(X) := \sum_{i \in X}w_i$, which is extended
to set families $\Ff \subseteq 2^{[n]}$ by denoting $w(\Ff) := \{ w(X) : X \in \Ff\}$. We use $A \uplus B = C$ to denote that $A,B$ form a partition of $C$.

\paragraph{Prime numbers and hashing}

We use the following folklore theorem on prime numbers:

\begin{lemma}[Folklore]
	\label{lem:pnt}
	For every sufficiently large integer $r$ the following holds. If $p$ is a prime between $r$ and $2r$
	selected uniformly at random and x is a nonzero integer, then p divides x with probability at most $(\log_2 x)/r$.
\end{lemma} 

The following Lemma already appeared in~\cite{stacs2016}, but
since we need slightly different parameters we repeat its proof.

\begin{lemma}[cf., Proposition 3.5 in \cite{stacs2016}]
    \label{lem:epsmixer}

    Let $w_1,\ldots,w_n$ be $n$ integers bounded by $2^{\Oh(n)}$.
    Suppose $Q \subseteq [n]$ with $|Q| = \Theta(n)$. Let $W_1,\ldots,W_c$ be integers and let $W=\prod_{i=1}^c W_i$ such that $W \le |w(2^Q)|$. For $i=1,\ldots,c$, let
    $p_i$ be prime numbers selected uniformly at random from
    $[W_i/2,W_i]$.

	Let $s_0$ be the smallest integer such that $\binom{|Q|}{s_0}
    \geq |w(2^Q)|/|Q|$. Denoting $p:=\prod_{i=1}^c p_i$, we have
	\[
        \Pr\Big[ \Big|\Big\{ a\; \mathrm{ mod }\; p \; : \; X \subseteq Q,\; |X| \in
                [s_0,|Q|/2],\; w(X) = a \Big\}\Big| \geq
                \Omega\left(\frac{p}{n^c}\right)
                 \Big] \geq  9/10.
	\]
\end{lemma}

\begin{proof}[Proof of Lemma~\ref{lem:epsmixer}]
    By the pigeonhole principle, there exists an integer $s_1$ such that
    $\left|w\left(\binom{Q}{s_1}\right)\right| \geq |w(2^Q)|/|Q|$. By the
    minimality of $s_0$ we know that $s_1 \geq s_0$.  We may also assume
    $s_1 \leq |Q|/2$ because $w\left(\binom{Q}{s}\right)=w\left(\binom{Q}{|Q|-s}\right)$ by subset
    complementation.
	
    Let $\Ff \subseteq \binom{Q}{s_1}$ be a maximal injective subset, i.e.,
    satisfying $|\Ff| = |w(\Ff)| = |w\left(\binom{Q}{s_1}\right)| \ge
    |w(2^Q)|/|Q|$.  Let $c_i = |\left\{ Y \in \Ff: w(Y) \equiv_p
    i\right\}|$ be the number of sets from $\Ff$ with a sum in the $i$'th congruence class modulo $p$.  Our
    goal is to lower bound the probability that $c_{i} > 0$ for a random $i \in \mathbb{Z}_p$.  We can bound the
    expected $\ell^2$ norm (e.g., the number of collisions) by

	\begin{equation}
	\label{eqn:alg1 collision prob}
        \Ex{ \sum_{i} c^2_i } = \sum_{Y,Z \in \Ff}\prob{p \text{ divides } w(Y)-w(Z)}{} \leq |\Ff| + \Oh(n^c|\Ff|^2 /p).
	\end{equation}

    The last inequality follows by Lemma~\ref{lem:pnt} and the assumption that
    $w_1,\ldots,w_n \le 2^{\Oh(n)}$ and $|Q| = \Theta(n)$.
    Namely, note that if $Y \neq Z$, then $w(Y)\neq w(Z)$. Hence $\Pr[p \text{
    divides } w(Y)-w(Z)]$ is at most $\Oh(n^c/W)$ by applying Lemma~\ref{lem:pnt} $i$ times with each $p_i$.

    By Markov's inequality, $\sum_{i} c^2_i \le \Oh(|\Ff| + n^c|\Ff|^2/p)$ with
    constant probability over the choice of $p$. We assumed that $p \le \Oh(w(2^Q))$,
    hence $|\Ff| \le \Oh(n^c |\Ff|^2/p)$. So $\sum_i c^2_i \le \Oh(n^c|\Ff|^2/p)$.

    Conditioned on this, the Cauchy-Schwarz inequality implies that the number of
    non-zero $c_i$'s is at least $|\Ff|^2\big/\sum_{i} c_i^2 \ge
    \Omega(\min\{|\Ff|, p/n^c\}) = \Omega(p/n^c)$ as desired.
\end{proof}

\paragraph{Shroeppel-Shamir's sumset enumeration}
We recall some of the basic building blocks of previous work on Subset
Sum. In~\cite{schroeppel} the authors used the following data structure to obtain an $\Ot(n^2)$ time and $\Ot(n)$ space algorithm for 4-SUM.

\begin{restatable}{lemma}{incdatastructure}
    \label{lem:sumsetenum}

    Let $A,B \subseteq \mathbb{Z}$ be two sets of integers, and let $C := A+B :=
    \{a+b : a\in A, b\in B \}$ be their sumset. Let $c_1,\ldots,c_m$ be
    elements of $C$ in increasing order. There is a data structure $\inc:=
    \inc(A,B)$ that takes $\Ot(|A|+|B|)$ preprocessing time and supports a query
    $\inc.\nextel$ that in the $i$'th (for $1 \leq i \le m$) call outputs
    $(P^l_{c_i},c_i)$, and in the subsequent calls outputs $\mathtt{EMPTY}$. Here
    $P^l_{c_i}$ is the set $\{ (a,b) : a \in A, b \in B, a+b = c_i \}$.

    Moreover, the total time needed to execute all $m$ calls to $\inc.\nextel$ is
    $\Ot(|A||B|)$ and the maximum space usage of the data structure is $\Ot(|A| + |B|)$.

    Similarly, there is a data structure $\dec:= \dec(A,B)$ that outputs pairs of
    elements of $A$ and $B$ in order of their decreasing sum.
\end{restatable}

The data structure crucially relies on priority queues. We included
the proof of this Lemma in Appendix~\ref{proof-of-lem:sumsetenum} for
completeness.

\paragraph{Binomial Coefficients}
We will frequently use the binary entropy function $h(p):= -p \log_2(p)-(1-p)\log_2(1-p)$. Its main use is via the following estimate of binomial coefficients: 
\begin{equation}\label{eq:binent}
	\Omega(d^{-1/2}) 2^{dh(\alpha)} \le \binom{d}{\alpha d} \le 2^{dh(\alpha)}.
\end{equation}

We also consider the inverse of the binary entropy. Since $h(\alpha)$ is
strictly increasing in $[0,0.5]$ we can define $h^{-1} : [0,1] \rightarrow [0,0.5]$,
with condition that $h^{-1}(\alpha) = \beta$ iff $h(\beta) = \alpha$.

For every $\alpha \in [0,0.5]$ we have the following inequality on the entropy function:
\begin{equation}
\label{ineq:entropy}
1 - 4\alpha^2 \leq h(1/2 - \alpha) \le 1-2\alpha^2/\ln{2}
\end{equation}
We can also compute the derivative of the entropy function on $1/4$ to bound its value, i.e., for every $\alpha>0$:
\begin{equation}
\label{ineq:entropy2}
h(1/4 + \alpha) \le h(1/4) + \alpha \log_2{3}
\end{equation}
Moreover by the concavity of binary entropy we know that for all $\alpha,x,y \in
[0,1]$:
\begin{equation}
\label{ineq:convexity}
\alpha h(x) + (1-\alpha) h(y) \le h(\alpha x + (1-\alpha) y)
\end{equation}
In particular it means that $h(\sigma\lambda) + h((1-\sigma)\lambda) \le
2h(\lambda/2)$ for any $0\leq \sigma \leq 1$.

The following standard concentration lemma will be useful to control the intersection of the solution with certain subsets of the weights of the subset sum instance:
\begin{lemma}\label{lem:conc}
Let $A \subseteq [d]$ be any set with $|A| =
\alpha d$, and let $B \subseteq
[d]$ be uniformly sampled over all subsets with $|B| = \beta d$ and $\alpha
\beta d$ be an integer. Then the following holds:
	\begin{displaymath}
	\prob{|A \cap B| = \alpha\beta d}{} \ge \Oms(1).
	\end{displaymath}
\end{lemma}
\begin{proof}
There are $\binom{d}{|B|}$ possibilities of selecting a random $B$.
There are $\binom{|A|}{\alpha\beta d} \binom{d-|A|}{|B| -
	\alpha\beta d}$ many possibilities of
selecting $B$, such that $|A \cap B| = \alpha \beta d$. 
Hence for a random $B$, the probability that $|A \cap B| = \alpha \beta d $
is:

\begin{displaymath}
\frac{\binom{\alpha d}{\alpha \beta d}
	\binom{(1-\alpha)d}{\beta(1-\alpha)d}}{\binom{d}{\beta d}} \ge 
\Omega\left(\frac{2^{d(\alpha h(\beta) + (1-\alpha)h(\beta))})}{d
	2^{dh(\beta)}}\right)
= \Omega\left(\frac{1}{d}\right),
\end{displaymath}
because of~\eqref{eq:binent}.
\end{proof}

\paragraph{Preprocessing Algorithms}
\label{sec:prepro}

We now present several simple procedures that allow us to make
assumptions about the given Subset Sum instance in the proof of Theorem~\ref{main-thm}.
Throughout this paper $w_1,\ldots,w_n,t$ denotes an instance of Subset Sum. We
can assume that the integers $w_1,\ldots,w_n,t$ are positive and $w_1+ \ldots  + w_n
+ t \leq 2^{10 n}$ (see \cite[Lemma 2.1]{stacs2016}). Throughout the paper we will introduce certain constants close to $0$ and assume that $n$ is big enough, so the product
of $n$ with these constants is an integer.

The following notion that was already discussed in Section~\ref{sec:ss-via-ov}
corresponds to the number of distinct sums of the subsets of a given set.

\begin{definition}[$\eps$-mixer] \label{def:epsmix}
    A set $M \subseteq [n]$ is an \emph{$\eps$-mixer} if $|w(2^M)| = 2^{(1-\eps)|M|}$.
\end{definition}

\begin{lemma}
	\label{mixation-parameter}
	Given a set $M$, one can in $\Os(2^{|M|})$ time and $\Os(2^{|M|})$ space determine the $\eps$ such that $M$ is an $\eps$-mixer.
\end{lemma}
\begin{proof}
	Iterate over every possible subset of $M$ and store $w(2^{M})$. Afterwards sort $w(2^{M})$, determine the size of $M$ and output $\eps := (1-\log_2 (|w(2^{M})|)/|M|)$.
\end{proof}

\begin{lemma} 
    \label{win-win}
    For any constants $\eps_0 > 0$ and $\mu \in (0,1/4)$, there is an algorithm
    that, given a Subset Sum instance
    $w_1,\ldots,w_n,t$ and an $\eps$-mixer $M$ satisfying $|M| = \mu n$ and $\eps > \eps_0$, solves the instance in time $\Os(2^{(1 - \eps_0\mu)n/2})$ and $\Os(2^{(1 -
    \eps_0\mu)n/4})$ space.
\end{lemma}

\begin{proof}[Proof of Lemma~\ref{win-win}]
    Arbitrarily partition $[n]\setminus M = L_1 \uplus  L_2 \uplus R_1 \uplus R_2$,
    such that:
    \begin{align*}
        |L_1| = & \frac{n}{4} - \mu n\left(1-\frac{3}{4}\eps_0\right), &&& |L_2| = |R_1| = |R_2| = & \frac{(1 - \eps_0\mu)n}{4}.
    \end{align*} 
    Observe that $|L_1| > 0$, because $\mu < 1/4$ and $\eps_0 > 0$.
    Then construct the sets $A: = w(2^{L_2}),B:= w(2^{R_1}),C :=w(2^{R_2})$. Observe that the space
    needed to do this is exactly $2^{|L_2|}$, which is within the boundaries of
    our algorithm.  Next construct $D: = w(2^{L_1 \cup M})$. Observe that 
    \begin{displaymath}
        |D| = |w(2^{L_1 \cup M})| \le |w(2^{L_1})| |w(2^M)| \le 2^{(1 - \eps_0\mu)n/4}.
    \end{displaymath}
    Finally, observe that the Subset Sum instance is equivalent to the 4-SUM instance $A,B,C,D,t$, which we can solve in
    $\Os(|A||B| + |C||D|)$ time and $\Os(|A|+|B|+|C|+|D|)$ space using Lemma~\ref{ref:4sum}.
\end{proof}

\begin{lemma} 
    \label{small-lambda}

    Suppose a Subset Sum instance $w_1,\ldots,w_n,t$ with promise that there is
    a solution of size $\lambda n$ is given. Then we can find $S \subseteq [n]$
    with $w(S) = t$ in randomized $\Os(2^{h(\lambda)n/2} + 2^{n/4})$ time and
    $\Os(2^{h(\lambda)n/4})$ space.
\end{lemma}

\begin{proof}[Proof of Lemma~\ref{small-lambda}]
    Let $S$ be the solution to the \ssum instance such that $|S| = \lambda n$.  Randomly partition
    $[n] = A_1 \uplus A_2 \uplus A_3 \uplus A_4$, each of size $n/4$.  By
    Lemma~\ref{lem:conc} with probability $\Oms(1)$ we have that $|A_i \cap X| = \frac{\lambda n}{4}$
    for all $i \in [4]$.
    Next for all $i \in [4]$ we enumerate sets:
    \begin{displaymath}
        \mathcal{A}_i = \left\{ w(B) \; \Big| \; B \subseteq
        \binom{A_i}{\lambda n/4} \right\}.
    \end{displaymath} 
    We can construct $\mathcal{A}_i$ in $\Os(2^{n/4})$ time and
    $\Os(|\mathcal{A}_i|)$ space by testing all possible subsets of $A_i$.
    Finally, we invoke 4-SUM algorithm from Lemma~\ref{ref:4sum} on instance
    $\mathcal{A}_1,\ldots,\mathcal{A}_4,t$. It runs in $\Os\left(
    { {n/4} \choose {\ell/4} }^2 \right)$ time and $\Os\left( { {n/4} \choose
    {\ell/4} } \right)$ space. For correctness, observe that $|\mathcal{A}_i| = \Os({{n/4} \choose {\lambda
    n /4}})$ and with $\Oms(1)$ probability $w(S\cap A_i) \in \mathcal{A}_i$.
\end{proof}

\section{Improving Schroeppel and Shamir: Proof of Theorem~\ref{main-thm}}
\label{main-proof}

This section is devoted to the proof of Theorem~\ref{main-thm}.  The main
technical effort, done in Subsections~\ref{sec:algo} to~\ref{subsec:runtime}, is
to prove the following lemma. 

\begin{lemma}[Main Lemma]\label{main-lem}
    Let $\lambda_0 := 0.495$, $\eps_0 := 0.00002$.
    Let $\lambda \in [\lambda_0,0.5]$, $\eps_R \in [0,\eps_0]$, $\mu \in (0.21,0.25)$ and let $M_L,M,M_R \subseteq [n]$ be disjoint sets such that $|M|=|M_L|=|M_R|=\mu n$.
    Let $0 \leq \eps \leq \eps_L \leq \eps_R$ be such that $M_L$ is an $\eps_L$-mixer, $M$ is an $\eps$-mixer and $M_R$ is an $\eps_R$-mixer.
	Let $S \subseteq [n]$ be such that $w(S)=t$ and $|M_L \cap S| = |M \cap S|=|M_R \cap S| = \lambda\mu n$.
	
	There is a Monte Carlo algorithm for Subset Sum that, given the instance $w_1,\ldots,w_n,t$, the sets $M_L,M,M_R$, and $\lambda,\eps_L,\eps_R$, runs in time $\Os(2^{n/2})$ and space

    \begin{displaymath}
        \Os\left(
            2^{(1/2 - \mu(3/2 + \lambda - h(1/4)))n+ 0.02\mu n} + 2^{\mu n(2h(1/4) - \lambda) + 0.02\mu n} + 2^{\mu n}
        \right).
    \end{displaymath}
\end{lemma}	

The performance of the algorithm depends on
the parameters $\lambda$, $\mu$, $\eps_L$ and $\eps_R$. It is
instructive to think about $\eps_L =\eps_R= 0$ and $\lambda=1/2$.

First, we prove the main result of the paper assuming Lemma~\ref{main-lem} by using the elementary
preprocessing algorithms provided in Section~\ref{sec:prepro}.

\begin{proof}[Proof Theorem~\ref{main-thm} assuming Lemma~\ref{main-lem}]
    Set $\mu := 0.217$. With polynomial overhead we can guess $|S| = \lambda n$.
    If $\lambda < \lambda_0$ then we use Theorem~\ref{small-lambda} to solve \ssum in
    $\Os(2^{h(\lambda)n/4}) \le \Os(2^{0.249982n})$ space and $\Os(2^{n/2})$ time. Hence, we can assume that
    $\lambda \ge \lambda_0$. We can also assume that $\lambda \leq 1/2$ by looking for $[n] \setminus S$ instead of $S$ by changing $t$ to $w([n])-t$.
	
    Next, randomly select pairwise disjoint sets $M,M_L,M_R \in { {[n]}
    \choose {\mu n}}$. For each of them we use Lemma~\ref{mixation-parameter} to
    determine the $\eps,\eps_L,\eps_R$ such that $M$ is an $\eps$-mixer, $M_L$
    is an $\eps_L$-mixer and $M_R$ is an $\eps_R$-mixer. If at least one of
    $\eps,\eps_L,\eps_R$ is at least $\eps_0$, use Theorem~\ref{win-win} to
    solve the instance in $\Os(2^{(1-\mu \eps_0)n/4}) \le \Os(2^{0.24999892})$
    space and $\Os(2^{n/2})$ time. Hence we can assume $\eps,\eps_L,\eps_R <
    \eps_0$.

    Finally, Lemma~\ref{main-lem} applies and it solves the instance in time
    $\Os(2^{n/2})$. For our choice of the parameters we get that the space is
    at most $\Os(2^{0.2491n})$. 
    
    In total, the space complexity of our algorithm is bounded by
    $\Os(2^{0.249999n})$ as claimed.
\end{proof}

The rest of this section is devoted to the proof of Lemma~\ref{main-lem}.
This lemma is an extension of Theorem~\ref{simple-secondmain-thm} combined with a fast OV algorithm. As mentioned
in Subsection~\ref{subsec:intui}, we apply the representation technique on 2
levels and therefore we need $3$ sets $M_L,M,M_R$. Moreover, the assumption $0<
\eps \leq \eps_L,\eps_R$ is to avoid the aforementioned undesired $\Os(2^{(0.5 +
\Oh(\eps))n})$ running time.

\subsection{The Algorithm for Lemma~\ref{main-lem}}
\label{sec:algo}
\begin{algorithm}
	\DontPrintSemicolon
	\SetKwInOut{Input}{Algorithm}
	\SetKwInOut{Output}{Output}
    \Input{$\mathtt{SubsetSum}(w_1,\ldots,w_n,t,M_L,M,M_R,\lambda,\eps_L,\eps_R$)}
	\Output{Set $S$ with $w(S)=t$ and $|M_L \cap S|,|M \cap S|, |M_R \cap S|=\lambda|M|$, if it exists}

	Arbitrarily partition $[n]\setminus (M_L \cup M \cup M_R)= L \uplus R$ such that $|L|,|R|$ satisfy~\eqref{eq:lr}\label{lin:lr}
	
    Pick random primes $p_R \in \Theta(2^{(\lambda-\eps_R)|M|})$, $p' \in \Theta(2^{(\eps_R-\eps_L)|M|})$; set $p_L=p' \cdot p_R$ \\
    Pick random $x_L \in \mathbb{Z}_{p_L}, x \in \mathbb{Z}_{p_L}, x_R \in \mathbb{Z}_{p_R}$ \\
    \ForEach{$\sigma,\sigma_L,\sigma_R$ s.t. $h(\sigma),h(\sigma_L) \geq 1-\eps_L/\lambda-\frac{\log_2
n}{n}$ \textnormal{and} $h(\sigma_R) \geq  1-\eps_R/\lambda-\frac{\log_2
n}{n}$ }{\label{lin:sigmaiter}
		Construct $\Ll_1,\Ll_2,\Rr_1,\Rr_2$ as defined in \Crefrange{eq:l1}{eq:r2}\\
        \If{$\mathtt{WeightedOV}(\Ll_1,\Ll_2,\Rr_1,\Rr_2,M,t)$}{\label{lin:retsol}
            \textbf{return} true
        }
		
	}
    \textbf{return} false
	\caption{Pseudocode of the algorithm for Lemma~\ref{main-lem}}
	\label{main-alg}
\end{algorithm}

Algorithm~\ref{main-alg} presents the pseudocode of
Lemma~\ref{main-lem}.  The $\mathtt{WeightedOV}$ subroutine decides whether
there exists $(A_1,\ldots,A_4) \in \Ll_1 \times \Ll_2 \times \Rr_1 \times \Rr_2$
with $w(A_1\cup \ldots\cup A_4) = t$ and $A_i \cap A_j =\emptyset$ for all $i
\neq j$. This subroutine will be provided and analysed later in the
Section~\ref{subsec:runtime}.

On a high level, Algorithm~\ref{main-alg}  has the same structure as
Algorithm~\ref{secondmain-alg}, with one major difference: The sets $\Ll$ and
$\Rr$ are generated implicitly. To generate these lists we combine the technique
from~\cite{schroeppel} as summarized in Lemma~\ref{lem:sumsetenum} with two
more applications of the representation technique used to generate $\Ll$ and
$\Rr$.\footnote{Note that, formally speaking, the list $\Ll$ from
Algorithm~\ref{secondmain-alg} is not the same as the set of elements of list
$\Ll$ of Algorithm~\ref{main-alg}, but since the two are almost
identical we kept the same notation. }

The algorithm iterates over every possible choice of parameters
$\sigma,\sigma_L,\sigma_R\in [0,1]$, such that
$h(\sigma),h(\sigma_L) \ge 1-\eps_L/\lambda-\frac{\log_2 n}{n}$ and $h(\sigma_R) \ge 1-\eps_R/\lambda-\frac{\log_2 n}{n}$ in
Line~\ref{lin:sigmaiter}. The precision of $\sigma,\sigma_R,\sigma_L$ is
polynomial, since these parameter describe the size of possible subsets of
$M,M_R,M_L$. The purpose of one iteration of this loop is
summarized in the following lemma, which is also illustrated in
Figure~\ref{fig:realpart}:

\begin{lemma}\label{lem:detect}
	Consider an iteration of the loop at Line~\ref{lin:sigmaiter} of Algorithm~\ref{main-alg} with parameters $\sigma,\sigma_L,\sigma_R$. Suppose there exists a set $S \in \binom{[n]}{\lambda n}$ with $w(S)=t$ that has a partition $S= S_1 \uplus S_2 \uplus \cdots \uplus S_8$ satisfying the following properties:
	\begin{gather}\label{eq:conditions}
		\begin{align}
			S_1 &\subseteq L, 					& S_2 &\in \binom{M_L}{\sigma_L\lambda|M|}, &  S_3 &\in \binom{M_L}{(1-\sigma_L)\lambda|M|}, & S_4 &\in \binom{M}{\sigma\lambda|M|},\nonumber\\
			 S_8 &\subseteq R, 					& S_7 &\in \binom{M_R}{\sigma_R\lambda|M|}, &  S_6 &\in \binom{M_R}{(1-\sigma_R)\lambda|M|}, & S_5 &\in \binom{M}{(1-\sigma)\lambda|M|},
		\end{align}\\
        w(S_1 \cup S_2 \cup S_3 \cup S_4) \equiv_{p_L} x, \quad w(S_5 \cup S_6 \cup S_7 \cup S_8) \equiv_{p_R} t-x, \nonumber\\
        w(S_1 \cup S_2) \equiv_{p_L} x_L, \quad   w(S_3 \cup S_4) \equiv_{p_L} x - x_L, \nonumber\\
        w(S_5 \cup S_6) \equiv_{p_R} x_R, \quad   w(S_7 \cup S_8) \equiv_{p_R} t -x - x_R. \nonumber
	\end{gather}
    Then during this iteration the Algorithm~\ref{main-alg} returns true.
\end{lemma}

\begin{figure}
	\center
	\includegraphics[width=0.75\textwidth]{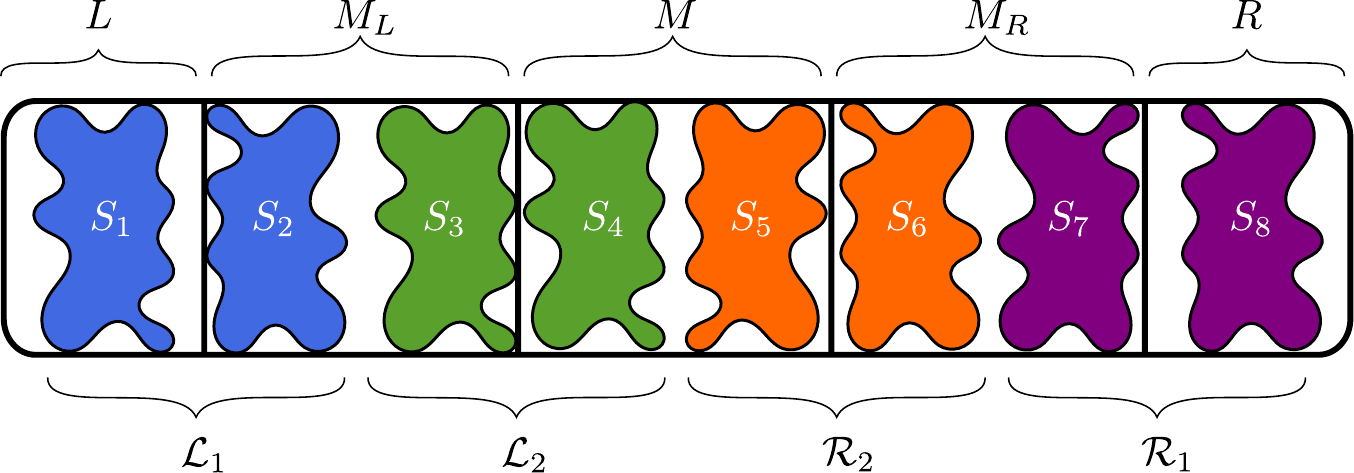}
	\caption{ The decomposition of the solution $S=S_1\uplus\ldots \uplus S_8$ as formalized in Lemma~\ref{lem:detect}.}
	\label{fig:realpart}
\end{figure}

The (relatively straightforward) proof of Lemma~\ref{lem:detect} will be given in Subsection~\ref{subsec:correctmainalg} where we prove the correctness of the algorithm.
To obtain a relatively fast algorithm in the case that $\eps$ is bounded
away from $0$ or $\lambda$ is bounded away from $1/2$, we need to carefully
define the lists $\Ll_1,\Ll_2,\Rr_1,\Rr_2$ in order to not slow down the run
time to beyond $\Os(2^{n/2})$. To do so, we use the following balance parameter 
\[
	\beta = \beta(\lambda,\sigma) := h(\sigma\lambda) - h((1-\sigma)\lambda).
\]
Intuitively, $\beta$ expresses the difference of the expected list sizes
$\{\Ll(a)\}_{a}$ and $\{\Rr(b)\}_{b}$ (see
Line~\ref{lin:constructll} and Line~\ref{lin:constructrr} of
Algorithm~\ref{main-weighted-OV}) when we would have set $|L|=|R|$.
Observe that if $\eps=0$ and $\lambda=1/2$, then $\sigma_L,\sigma,\sigma_R=1/2$ and indeed $\beta=0$.

All elements of $[n]$ not in $M_L \cup M \cup M_R$ are arbitrarily partitioned into $L$ and $R$ on Line~\ref{lin:lr} where $|L|$ and $|R|$ are chosen to compensate for imbalance caused by $\eps,\sigma,\lambda$ as follows:
\begin{align}
\label{eq:lr}
    |L| & = \frac{(1 - 3\mu - \beta\mu) n}{2}, & & |R|  = \frac{(1 - 3\mu +
    \beta\mu) n}{2}.
\end{align}
Observe that $|\beta| \le 1$, and since $\mu \leq 1/4$ we have that $|L|,|R| > 0$.

Now we define the four lists that play a similar role in our algorithm as the four lists in the original algorithm of~\cite{schroeppel}.
\begin{align}
\label{eq:l1}
    \Ll_1 & := \bigg\{ S_1 \cup S_2 \; \bigg| \; w(S_1 \cup S_2) \equiv_{p_L} x_L,&& S_1 \subseteq L,&& S_2 \in \binom{M_L}{\sigma_L \lambda|M|} \bigg\}, \\
   \label{eq:r1}
    \Rr_1 & := \bigg\{ S_7 \cup S_8 \; \bigg| \; w(S_7 \cup S_8) \equiv_{p_R} t-x-x_R,&& S_8 \subseteq R,&& S_7 \in\binom{M_R}{\sigma_R\lambda|M|}\bigg\}, 
\end{align}

\begin{align}
\label{eq:l2}
\Ll_2 &:= \bigg\{ S_3 \cup S_4 \; \bigg| \; w(S_3 \cup S_4) \equiv_{p_L} x-x_L,  
&&S_3 \in \binom{M_L}{(1-\sigma_L)\lambda|M|},
&&S_4 \in \binom{M}{\sigma\lambda|M|} \bigg\},\\
\label{eq:r2}
\Rr_2 &:= \bigg\{ S_5 \cup S_6 \; \bigg| \;  w(S_5 \cup S_6) \equiv_{p_R} x_R, 
&&S_5 \in \binom{M}{(1-\sigma)\lambda|M|}, 
&&S_6 \in \binom{M_R}{(1-\sigma_R)\lambda|M|} \bigg\}.
\end{align}

Using a straightforward algorithm, we can construct each list using
$\Ot(|\Ll_1|+|\Ll_2|+|\Rr_1|+|\Rr_2|+2^{\mu n})$ time and space (see
Lemma~\ref{enumeration-lemma} in Appendix~\ref{enum-appendix}).

\subsection{The Weighted Orthogonal Vectors Subroutine}
\label{subsec:runtime}

\begin{algorithm}
	\DontPrintSemicolon
	\SetKwInOut{Input}{Algorithm}
	\SetKwInOut{Output}{Output}
	\Input{$\mathtt{WeightedOV}$($\Ll_1,\Ll_2,\Rr_1,\Rr_2,M,t$)}
	\Output{If $\exists$ disjoint $(A_1,\ldots,A_4) \in \Ll_1 \times\Ll_2 \times \Rr_2 \times \Rr_1$ with $w(A_1 \cup \ldots \cup A_4) = t$ }
	
	Initialize $\inc= \inc(w(\Ll_1),w(\Ll_2))$ \tcp*{Lemma~\ref{lem:sumsetenum}}
	Initialize $\dec = \dec(w(\Rr_1),w(\Rr_2))$ \tcp*{Lemma~\ref{lem:sumsetenum}}
	Let $(P^r_b,b) = \dec.\nextel$\\
	\ForEach(\tcp*[f]{Integers $a$ are increasing}){$(P^l_a,a) = \inc.\nextel$}{\label{lin:linsearch}
		\While {$a+b > t$} {
			Let $(P^r_b,b) = \dec.\nextel$ \tcp*{Integers $b$ are decreasing}
		}
		\If {$a+b = t$}{
			Construct $\mathcal{L}(a):= \Big\{ Y \cap M \ \Big|\  \exists X \in \Ll_1, Y \in \Ll_2, X \cap Y = \emptyset, (w(X),w(Y)) \in P^l_a \Big\}$\label{lin:constructll}\\
			Construct $\mathcal{R}(b):= \Big\{ Y \cap M \ \Big|\  \exists X \in \Rr_1, Y \in \Rr_2,  X \cap Y = \emptyset, (w(X),w(Y)) \in P^r_b \Big\}$\label{lin:constructrr}\\
			\LineIf{$\mathtt{OV}(\Ll(a),\Rr(b))\neq \emptyset$}{\textbf{return } true\label{lin:ov}} \tcp*{Theorem~\ref{ov-lemma}}
		}
	}
	\textbf{return} false
	\caption{Weighted Orthogonal Vectors algorithm}
	\label{main-weighted-OV}
\end{algorithm}

Now we describe the $\mathtt{WeightedOV}$ subroutine (see pseudo-code in Algorithm~\ref{main-weighted-OV}). The
algorithm is heavily based on the data structures from~\cite{schroeppel} as
described in Lemma~\ref{lem:sumsetenum}. First we initialize the queue $\inc$
for enumerating $w(\Ll_1) + w(\Ll_2)$ in the increasing order and the queue $\dec$ for enumerating 
$w(\Rr_1) + w(\Rr_2)$ in the decreasing order. 
With these queues, we enumerate all groups $\Ll(a) \subseteq M$ with the property that if $S_4 \in
\Ll(a)$ then there exist $X \in \Ll_1$ and $Y \in \Ll_2$ with $Y\cap M = S_4$, $X \cap Y = \emptyset$ and $w(X) +
w(Y) = a$.
Similarly, we enumerate all groups $\Rr(a) \subseteq M$ with the property that if $S_5 \in
\Rr(b)$ then there exist $X \in \Rr_1$ and $Y \in \Rr_2$ with $Y\cap M = S_5$, $X \cap Y = \emptyset$ and $w(X) +
w(Y) = b$.
In the end we
execute a Monte-Carlo algorithm $\mathtt{OV}$ that solves the unweighted orthogonal vectors problem that will be described in  Theorem~\ref{ov-lemma}.

We now analyse the correctness and space usage of this algorithm. The time analysis will be intertwined with the time analysis of Algorithm~\ref{main-alg} and is therefore postponed to Subsection~\ref{mainalg:time}.

\begin{lemma}
 Algorithm $\mathtt{WeightedOV}$ is a correct Monte-Carlo algorithm for the Weighted Orthogonal Vectors Problem.
\end{lemma}
\begin{proof}
    If the algorithm outputs true at Line~$\ref{lin:ov}$, there exist $A_1 \in
    \Ll_1,A_2 \in \Ll_2 ,A_3 \in \Rr_2, A_4 \in \Rr_1$ such that
    $w(A_1)+w(A_2)+w(A_3)+w(A_4)=t$. 
    
    First, note that by the construction of sets $\Ll_1,\Ll_2,\Rr_1,\Rr_2$ it has
    to be that $A_2 \cap A_3 \subseteq M$.  Since the $\mathtt{OV}$ algorithm checks
    for disjointness on $M$ we have that $A_2 \cap A_3 \cap M = \emptyset$,
    hence $A_2 \cap A_3 = \emptyset$.  Also, $A_1 \cap A_2 = \emptyset$ because $(X,Y) \in \Ll(a)$ means $X \cap Y =
    \emptyset$. Similarly $A_3 \cap A_4 = \emptyset$ because $(X,Y) \in \Rr(b)$
    means that $X \cap Y = \emptyset$ . By the construction of
    the lists $\Ll_1,\Ll_2,\Rr_1,\Rr_2$ the sets $A_1,\ldots,A_4$ are thus
    mutually disjoint and indeed the instance of Weighted Orthogonal Vectors is
    a YES-instance.
	
    For the other direction, assume the desired $A_1,\ldots,A_4$ quadruple
    exists. Let $t_L := w(A_1 \cup A_2)$.  Then $t_L \in w(\Ll_1)+w(\Ll_2)$ and
    $t-t_L = w(A_3 \cup A_4) \in w(\Rr_1)+w(\Rr_2)$.  By
    Lemma~\ref{lem:sumsetenum} $\inc$ enumerates $w(\Ll_1)+w(\Ll_2)$, and $\dec$
    enumerates $w(\Rr_1)+w(\Rr_2)$ in decreasing order.  Therefore, since the loop
    starting at Line~\ref{lin:linsearch} is a basic linear search routine, it
    sets $a$ to $t_L$ and $b$ to $t-t_L$ in some iteration: If $a$ is set to
    $t_L$ before $b$ is set to $t-t_L$, then $b$ is in this iteration larger
    than $t-t_L$ and it will be decreased in the next iterations until it is set
    to $t-t_L$.  Similarly, if $b$ is set to $t-t_L$ before $a$ is set to $t_L$,
    in this iteration $a$ is smaller than $t_L$ and it will be increased in the
    next iterations until it is set to $t_L$.
	
    In the iteration with $a = t_L$ and $b = t - t_L$ we have that $P^l_a$
    contains the pair $(w(A_1), w(A_2))$ and $P^r_b$ contains the pair $(w(A_4),
    w(A_3))$.  Therefore $\Ll(a)$ contains $A_2 \cap M=S_4$ and $\Ll(b)$
    contains $A_4 \cap M=S_5$. Since $S_4$ and $S_5$ are disjoint a
    solution will be detected by the $\mathtt{OV}$ subroutine with at least
    constant probability on Line~\ref{lin:ov}. 
\end{proof}

\begin{lemma}
	 Algorithm $\mathtt{WeightedOV}$ uses at most $\Os(|\Ll_1| + |\Ll_2| + |\Rr_1| + |\Rr_2| + 2^{|M|})$ space.
\end{lemma}	
\begin{proof}
	The datastructures $\inc$ and $\dec$ use at most $\Ot(|\Ll_1| + |\Ll_2| +
    |\Rr_1| + |\Rr_2|)$ space by Lemma~\ref{lem:sumsetenum}, and the sets $\Ll(a)$ and $\Rr(b)$ are of cardinality at most $2^{|M|}$. 
	The statement follows since, as we will show in Theorem~\ref{ov-lemma}, the subroutine $\mathtt{OV}(\Aa,\Bb)$ uses at most $\Ot(|\Aa|+|\Bb|+2^{|M|})$ space.
\end{proof}

\subsection{Correctness of Algorithm~\ref{main-alg}}
\label{subsec:correctmainalg}

We now focus on the correctness of the entire algorithm.
First notice that if the algorithm finds a solution on Line~\ref{lin:retsol}, it is always correct since it found pairwise disjoint sets $A_1,A_2,A_3,A_4$ satisfying $w(A_1 \cup A_2 \cup A_3 \cup A_4)=t$. Thus $S:= A_1 \cup A_2 \cup A_3 \cup A_4$ is a valid solution. The proof of the reverse implication is less easy and its proof is therefore split in two parts with the help of Lemma~\ref{lem:detect}.

Note that because the partition $[n] = L \uplus M_L \uplus M \uplus M_R \uplus R$ is selected
at random, the solution is well-balanced in sets $L,M_L,M,M_R,R$.
The following is a direct consequence of Lemma~\ref{lem:conc}:
\begin{observation}
	\label{good-split}
    Let $S$ be the solution to the \ssum instance with $|S| = \lambda n$. Then,
	with $\Oms(1)$ probability, the following holds:
	$|S \cap M_L| = \lambda |M_L|, \;|S \cap M| =  \lambda |M|, \; |S \cap M_R| = \lambda |M_R|. $
\end{observation}

Now we show that if the above event was successful, the conditions of Lemma~\ref{lem:detect} apply with good probability:

\begin{lemma}\label{lem:corprobbound}
	Suppose there exists a solution $S \subseteq [n]$ be such that $w(S)=t$ and $|M_L \cap S| = |M \cap S|=|M_R \cap S| = \lambda\mu n$. Then with probability $\Oms(1)$, there exists a partition $S= S_1 \uplus \cdots \uplus S_8$ satisfying all conditions in~\eqref{eq:conditions}.
\end{lemma}
\begin{proof}

    We select $S_1 = L \cap S$, $S_8 = R \cap S$, and $a,b$ be such that let $a \equiv_{p_L} w(S_1)$ and $b \equiv_{p_R} w(S_8)$. 
    Next we prove that, because the subsets of $M$ generate many distinct sums, the same holds for the solution intersected with $M$:
    \begin{claim}
        \label{mix-claim}
        The set $M \cap S$ is an $\eps'$-mixer for some $\eps'\leq \eps/\lambda$. Similarly, $M_L \cap S$ is
        an $\eps'_L$-mixer for $\eps'_L \leq \eps_L/\lambda$, and $M_R \cap S$ is an $\eps'_R$-mixer for some $\eps'_R \leq \eps_R/\lambda$.
    \end{claim}
    \begin{proof}[Proof of Claim~\ref{mix-claim}]
        Let us focus on $M \cap S$ (the result for $M_L$ and $M_R$ is analogous).
        Because $M$ is an $\eps$-mixer, we know that $2^{(1-\eps_L)|M|} \le |w(2^M)| \le |w(2^{M\cap
        S})| |w(2^{M\setminus S})|$. Since $|w(2^{M\setminus S})| \le
        2^{(1-\lambda)|M|}$ we have that $|w(2^{M\cap S})| \ge 2^{(\lambda -
        \eps)|M|} = 2^{(1-\eps/\lambda)|M \cap S|}$. 
    \end{proof}
    Now we know that $Q = M_L \cap S$ is a good mixer.  We can use
    Lemma~\ref{lem:epsmixer} for $Q = M_L \cap S$ and $p=p_L\cdot p'$, since $|w(2^{|M_L \cap S}|)|\geq 2^{(1-\eps_L/\lambda)|M_L \cap S|}= 2^{(\lambda-\eps_L)|M_L|}$. Because
    $x_L$ was chosen randomly, Lemma~\ref{lem:epsmixer} guarantees that with
    $\Oms(1)$ probability, there exists $S_2 \subseteq M_L \cap S$, such that
    $w(S_2) \equiv_{p_L} x_L - a$. Moreover Lemma~\ref{lem:epsmixer} guarantees that $|S_2| \in [s_0,\lambda \mu n/2]$, where
    $s_0$ is the smallest integer such that $\binom{Q}{s_0} \ge w(2^Q)/|Q|$.
    If we take the logarithm of both sides this is equivalent to
    \begin{displaymath}
        \lambda \mu n \cdot h\left( \frac{s_0}{|Q|} \right) \ge \log_2\left(
         \left|w(2^{(M_L
            \cap S)})\right|\right) - \frac{\log_2 n}{n} \ge (1-\eps_L/\lambda) \lambda \mu n
            - \frac{\log_2 n}{n}.
    \end{displaymath}

	Because we have checked all $\sigma_L$ that satisfy $h(\sigma_L) \ge
    (1-\eps_L/\lambda) - \frac{\log_2 n}{n}$ the algorithm will eventually guess the correct $s_0$ (and the same
    reasoning holds for $\sigma$ and $\sigma_R$). We select $S_3 = (M_L \cap S)
    \setminus S_2$ with $|S_3| = (1-\sigma_L)\mu n$.

    In a similar manner we can prove that with $\Oms(1)$ probability there exists $S_7 \subseteq M_R \cap
    S$, such that $w(S_7) \equiv_{p_R} (t-x-x_R) - b$ with $|S_7| = \sigma_R \mu n$ and
    $h(\sigma_R) \ge 1-\eps_R/\lambda - \frac{\log_2 n}{n}$ (we need to apply
    Lemma~\ref{lem:epsmixer} with $Q = M_R \cap S$ and prime $p_R$).
    Moreover, this probability only depends on $x_R$ which is independent of all other random variables and events.
    If this happens, we select $S_6 = (M_R \cap S) \setminus S_7$ with $|S_6| = (1-\sigma_R)\mu
    n$.

    Conditioned on the existence of $S_1,S_2,S_3,S_6,S_7,S_8$, now we prove there exist $S_4$ and $S_5$ with $\Oms(1)$ probability. Let $c = w(S_1\cup S_2 \cup S_3)$ and $d
    = w(S_6\cup S_7\cup S_8)$. We again use Lemma~\ref{lem:epsmixer}, but this time with $Q = M
    \cap S$ and $p=p_R\cdot p'$. It assures that with high probability there exist
    $S_4 \subseteq M \cap S$, with  $w(S_4) \equiv_{p_L} x-c$ and $|S_4| =
    \sigma \mu n$ with $h(\sigma) \ge 1-\eps_L/\lambda - \frac{\log_2 n}{n}$.
    And indeed, again this probability only depends on $x_R$ which is independent of all other random variables and events.
    If this event happens, we select $S_5 = (M\cap S) \setminus S_4$. 

    Now we use the fact that $p_R$ divides $p_L$: If $x \equiv_{p_L} a$ then $x \equiv_{p_R} a$ because $(x-a) = k\cdot p' \cdot p_R$ for some $k \in \mathbb{Z}$.
    Hence $w(S_5) + d \equiv_{p_R} w(S)-x$, which means that $w(S_5
    \cup S_6 \cup S_7 \cup S_8) \equiv_{p_R} t-x$. Moreover it holds that $|S_5| =
    (1-\sigma)\mu n$, thus $S_5$ also satisfies the desired conditions.

    To conclude observe that the constructed sets $S_1,\ldots,S_8$ are disjoint.
\end{proof}

Finally, we prove the Lemma~\ref{lem:detect}. Namely, we show that the existence of the tuple $(S_1,\ldots,S_8)$ implies that a solution is detected.

\begin{proof}[Proof of Lemma~\ref{lem:detect}]
	By the construction of $\Ll_1,\Ll_2,\Rr_1, \Rr_2$ and the assumed properties of the lemma, we have that $A_1:= S_1 \cup S_2 \in \Ll_1$, $A_2 := S_3 \cup S_4 \in \Ll_2$, $A_3:= S_5 \cup S_6 \in \Rr_2$, and $A_4:= S_7 \cup S_8 \in \Rr_1$. 
	Since the sets $S_1,\ldots,S_8$ are pairwise disjoint and satisfy
    $\sum_{i=1}^8 w(S_i)=t$, the sets $A_1,\ldots,A_4$ certify that
    $\mathtt{WeightedOV}(\Ll_1,\Ll_2,\Rr_1,\Rr_2,M,t)$ outputs true.
\end{proof}

The correctness of Algorithm~\ref{main-alg} directly follows by
combining Lemma~\ref{lem:corprobbound} and Lemma~\ref{lem:detect}.

\subsection{Space Usage of Algorithm~\ref{main-alg}}
\label{subsec:space}
The bulk of the analysis of the space usage consists of computing the expected
sizes of the lists $\Ll_1,\Ll_2,\Rr_1,\Rr_2$. This requires us to look closely
into the setting of the parameters.

\paragraph{Useful bounds on parameters}
Recall, that we defined the following constants
$\lambda_0 := 0.495$ and $\eps_0 := 0.00002$. Then, we assumed that
$\eps,\eps_L,\eps_R \le \eps_0$ and $\lambda \in [\lambda_0,0.5]$.
Moreover, we have chosen $\sigma,\sigma_L,\sigma_R$, such that:
\begin{align*}
    0.99995 < 1-\eps_0/\lambda_0 - \frac{\log_2{n}}{n} &\le h(\sigma),h(\sigma_L), h(\sigma_R)
\end{align*}

Which means that (for our choice of $\eps_0$ and $\lambda_0$ and large enough $n$):
\begin{align}
    \label{ineq:sigma}
    \sigma,\sigma_L,\sigma_R \in [0.495,0.505].
\end{align}

because $h(0.495) = h(0.505) \approx 0.999928$. Next, observe that
\begin{equation}
   \label{ineq:hsiglam}
   h(\sigma\lambda), h((1-\sigma)\lambda) \le h(1/4) + 0.004.
\end{equation}

because the entropy function is increasing in $[0,0.5]$ and $h(0.5 \cdot 0.505) - h(1/4) < 0.004$.
For the next inequality, recall that $\beta(\sigma,\lambda) = h(\sigma\lambda) -
h((1-\sigma)\lambda)$.
\begin{equation}\label{ineq:hbar}
	\begin{aligned}
        -0.012 \le \beta(\sigma,\lambda)  \le 0.012
	\end{aligned}
\end{equation}

because $|\beta| < h(0.505 \cdot 0.5) - h(0.495 \cdot \lambda_0) < 0.012$.

\paragraph{Bounds on the list sizes}
\begin{claim}
    \label{sizell1}
    $\Ex{|\Ll_1|}  \le \Os\left(2^{\left(1/2 - \mu(3/2 + \lambda - h(1/4) - 0.02) \right)n}\right)$.
\end{claim}
\begin{proof}
    Let $W_L$ be the number of possible different elements from $\Ll_1$. It is 
    \begin{displaymath}
        W_L := 2^{|L|} \binom{\mu n}{\lambda\sigma_L \mu n}
    \end{displaymath}
    The expected size of $\Ll_1$ over the random choices of $x_L$ is 
    \begin{displaymath}
        \Ex{|\Ll_1|} \le \frac{W_L}{p_L}.
    \end{displaymath}
    If we plug in the definition of $|L|$, we have:
    \begin{displaymath}
        (\log_2(\Ex{|\Ll_1|})/n) \le 1/2 - \mu (3/2 + \lambda -
            h(\lambda\sigma_L)) + \mu (\eps_L - \beta/2)
    \end{displaymath}
	By \eqref{ineq:hsiglam} we have that
    $h(\sigma_L\lambda) \le h(1/4) + 0.004$. By \eqref{ineq:hbar} we
    have that $|\beta| \le 0.012$ and $\eps_L < 0.01$. Hence:
    \begin{displaymath}
        (\log_2(\Ex{|\Ll_1|})/n) \le 1/2 - \mu (3/2 + \lambda - h(1/4)) + 0.02 \cdot \mu
        .
    \end{displaymath}
\end{proof}

By symmetry\footnote{The only difference being that $\beta$ shows up positively rather than negatively, but this does not matter since we bound its absolute value.} the same bound holds for $\Ex{|\Rr_1|}$.

\begin{claim}
    \label{sizerr1}
    $\Ex{|\Rr_1|}  \le \Os\left(2^{(1/2 - \mu(3/2 + \lambda - h(1/4) - 0.02))n}\right)$.
\end{claim}

Next we bound $|\Ll_2|$ and $|\Rr_2|$:

\begin{claim}
    \label{sizell2}
    $\Ex{|\Ll_2|}  \le \Os(2^{\mu n (2h(1/4) - \lambda) + 0.02 \mu n })$.
\end{claim}

\begin{proof}
    Let $W_L$ be the number of possibilities of selecting $S$. It is 
    \begin{displaymath}
        W_L := \binom{\mu n}{\sigma\lambda \mu n} \binom{\mu n}{(1-\sigma_L)\lambda \mu n}
    \end{displaymath}
    The expected size of $|\Ll_2|$ over the random choices of $x_L$ and $p_L$ is 
    \begin{displaymath}
        \Ex{|\Ll_2|} \le \frac{W_L}{p_L}.
    \end{displaymath}
    Hence,
    \begin{displaymath}
        (\log_2(\Ex{|\Ll_2|}))/n \le \mu(h(\lambda\sigma) + h(\lambda(1-\sigma_L)) - \lambda + \eps_L)
    \end{displaymath}
    We use Inequality~\ref{ineq:hsiglam} and have $h((1-\sigma)\lambda),h(\sigma\lambda) \le h(1/4) + 0.004$. Hence we
    can roughly bound:
    \begin{displaymath}
        (\log_2(\Ex{|\Ll_2|}))/n \le \mu (2h(1/4) - \lambda) + 0.02 \cdot \mu 
    \end{displaymath}
\end{proof}
By symmetry, the same bound holds for $|\Rr_2|$:
\begin{claim}
    \label{sizerr2}
    $\Ex{|\Rr_2|}  \le \Os(2^{\mu n (2h(1/4) - \lambda) + 0.02 \mu n })$.
\end{claim}

As mentioned in Subsection~\ref{sec:algo}, the subroutine  $\mathtt{WeightedOV}$ uses $\Os(|\Ll_1| + |\Ll_2| + |\Rr_1| + |\Rr_2| + 2^{|M|})$ space. By the above claims, we see that this is at most 

\begin{displaymath}
        \Os\left(
            2^{(1/2 - \mu(3/2 + \lambda - h(1/4)))n+ 0.02\mu n} + 2^{\mu n(2h(1/4) - \lambda) + 0.02\mu n} + 2^{\mu n}
        \right),
\end{displaymath}

as promised.

\begin{remark}\label{remark}
    The constant $0.02$ is based on our choice of $\eps_0$ and $\lambda_0$. When
    $\eps \rightarrow 0$ and $\lambda_0 \rightarrow 1/2$ it goes to $0$.  With
    more complicated inequalities and a tighter choice of parameters we were
    able to get $\Os(2^{0.249936n})$ space usage. We decided to skip the details
    for the simplicity of the presentation.
\end{remark}

\begin{remark}
    In this section we showed that expected sizes of $\Ll_1,\Ll_2,\Rr_1,\Rr_2$
    are bounded by $2^{(0.25-\delta)n}$ for some constant $\delta > 0$. With a
    standard Markov's inequality and union bound one can show that with
    $\Oms(1)$ probability it holds that sizes of $\Ll_1,\Ll_2,\Rr_1,\Rr_2$ are
    bounded by $2^{(0.25-\delta)n}$ for some constant $\delta > 0$.
\end{remark}

\subsection{Runtime of Algorithm~\ref{main-alg}} \label{mainalg:time}
Now, we prove that the runtime of Algorithm~\ref{main-alg} is $\Os(2^{n/2})$.
By Lemma~\ref{lem:sumsetenum}, the total runtime of all queries to \inc.\nextel is
$\Os(|\Ll_1||\Ll_2|)$, and the total runtime of all the
queries to \dec.\nextel is $\Os(|\Rr_1||\Rr_2|)$. This is upper bounded by
$\Os(2^{n/2})$ by the analysis of Subsection~\ref{subsec:space}.

The main bottleneck of the algorithm comes from all the calls to $\mathtt{OV}$ subroutine at Line~\ref{lin:ov} of Algorithm~\ref{main-weighted-OV}. To facilitate the analysis, we define sets $\mathcal{A},\mathcal{B}$ that represent the total input to the
$\mathtt{OV}$ subroutine: For every $a \in \nat$, such that $a \equiv_{p_L} x$ and
each $X \in \Ll(a)$, add the pair $(X \cap M,a)$ to $\mathcal{A}$ (without
repetitions). Similarly, for each $Y \in \Rr(t-a)$, add the pair $(Y \cap M,t-a)$ to
$\mathcal{B}$. Hence the total input for $\mathtt{OV}$ generated by is:
\begin{align*}
    \mathcal{A} &:= \biggl \{
        (X, a) :  X \in \binom{M}{\sigma\lambda\mu n} \text{ and } 
        a - w(X) \in w(2^{L\cup M_L}) \text{ and } a \equiv_{p_L} x
    \biggr\},\\
    \mathcal{B} &:= \biggl \{
        (Y, b) :  Y \in \binom{M}{(1-\sigma)\lambda\mu n} \text{ and }
        b - w(Y) \in w(2^{R\cup M_R}) \text{ and } b \equiv_{p_R} t-x
    \biggr\}
    .
\end{align*}
Now, let us calculate the expected size of $\mathcal{A}$. The number of possibilities of selecting possible elements in $\mathcal{A}$ is the number of
possibilities of selecting $X$ from $M$ and $a$ from $w(2^{L\cup M_L})$. Since the probability that $a \equiv_{p_L} x$ is $1/p_L$, we obtain
\[
    \Ex{|\mathcal{A}|} \le  \binom{M}{\sigma\lambda\mu n} |w(2^{L \cup M_L})| / p_L 
    .
\]
Similarly, the probability that $b \equiv_{p_R} t-x$ is $1/p_R$. To see this
recall that $x$ is chosen uniformly at random from $Z_{p_L}$, but since $p_L$ is
a multiple of $p_R$, integer $x \mod p_R$ is also uniformly distributed in $\mathbb{Z}_{p_R}$.

\[
    \Ex{|\mathcal{B}|} \le  \binom{M}{(1-\sigma)\lambda\mu n} |w(2^{R \cup M_R})| / p_R
    .
\]

Recall that $M_L$ is an $\eps_L$-mixer, hence $|w(2^{L\cup M_L})| \le |w(2^{|L|})|
2^{(1-\eps_L)\mu n}$, and similarly $M_R$ is an $\eps_R$-mixer. Hence:
\begin{align*}
    \log_2\left(\Ex{|\mathcal{A}|}\right) &\le |L| + (1-\eps_L)\mu n +
    h(\sigma\lambda)\mu n  - (\lambda - \eps_L)\mu n &\\
    &= \left( \frac{1-3\mu-\beta\mu}{2} + \mu- \lambda\mu +h(\sigma\lambda)\mu  \right) n, &\\
    &= \left(\frac{1}{2} - \mu \left( \frac{1}{2}+\lambda + \beta/2 - h(\sigma\lambda)  \right)  \right)  n,\\
  \intertext{and similarly:}
    \log_2\left(\Ex{|\mathcal{B}|}\right) &\le  |R| + (1-\eps_R)\mu n +
    h((1-\sigma)\lambda)\mu n  - (\lambda-\eps_R)\mu n\\
    &=\left( \frac{1-3\mu + \beta\mu}{2} + \mu -\lambda\mu + h((1-\sigma)\lambda)\mu  \right)n \\
    &=\left( \frac{1}{2} - \mu \left( \frac{1}{2} + \lambda -\beta /2 - h((1-\sigma)\lambda) \right)  \right)n.
\end{align*}
Now it becomes clear that we have chosen the balancing parameter $\beta$ in the sizes $|L|,|R|$ to match
the sizes of $\mathcal{A},\mathcal{B}$: Observe that
\[
	\beta/2-h(\sigma \lambda) = -\frac{h(\sigma\lambda)+h((1-\sigma)\lambda)}{2} = -\beta/2-h((1-\sigma)\lambda),
\]
and thus we obtain that

\[
 \log_2\left(\Ex{|\mathcal{A}|}\right),  \log_2\left(\Ex{|\mathcal{B}|}\right) \leq 
    \left(\frac{1}{2} - \mu\left(\frac{1}{2} + \lambda -
        \frac{h(\sigma\lambda)+h((1-\sigma)\lambda)}{2}\right) \right)n.
\]

By the concavity of binary entropy function (see~\eqref{ineq:convexity}), we know that $h(\sigma\lambda) +
h((1-\sigma)\lambda) \le 2h(\lambda/2)$. Hence:
\begin{equation}\label{eq:setsbounds}
    \Ex{|\mathcal{A}|} ,\Ex{|\mathcal{B}|} \le \Os(2^{n/2 - \mu n(1/2 + \lambda -
    h(\lambda/2))}).
\end{equation}

The $\mathtt{OV}$ subroutine (see Theorem~\ref{ov-lemma}) takes
$\mathcal{A}$ and $\mathcal{B}$ as an input with dimension $d = \mu n$. Note
that the condition $\lambda \in [0.4,0.5]$ in Theorem~\ref{ov-lemma} is satisfied by
the assumption in the Lemma~\ref{main-lem} and $\sigma \in [0.4,0.6]$ is
satisfied because for our choice of parameters $\sigma \in [0.495,0.505]$ (see~
\eqref{ineq:sigma}). 
Since the run time of the $\mathtt{OV}$ subroutine is linear in the input size, all calls to the $\mathtt{OV}$ algorithms jointly take the following total run time:

\begin{align*}
    \Os\left( 
        \left(|\mathcal{A}| + |\mathcal{B}|\right)
        2^{\mu n (1/2 + \lambda - h(\lambda/2))}
    \right).
\end{align*}

Thus the algorithm runs in $\Os(2^{n/2})$ time by~\eqref{eq:setsbounds}.

\begin{remark}
    Observe that the Inequality in Lemma~\ref{ovtime-inequality} is tight when $\lambda = 1/2$, which is the
    worst case for the algorithm. In particular any $\Os(2^{\delta d})$ improvement
    to our OV algorithm in the case $\lambda = 1/2$ and $\sigma=1/2$ for some $\delta >
    0$ would give an  $\Os(2^{(1/2 - \delta')n})$ time algorithm for Subset Sum for
    some $\delta' > 0$.
\end{remark}

\section{Reducing From Subset Sum to Exact Node Weighted $P_4$}
\label{sec:p4}

In this section we discuss a new connection between graph
problems and Subset Sum. Recall that in the Exact Node Weighted $P_4$ problem we
are given a node weighted graph $G=(V,E)$, and want to find $4$ vertices that form a path
and their total weight is equal $0$. 
We show that a fast algorithm for this problem would resolve Open Question~\ref{qmain}:

\weightedpfour*

\begin{proof}
	We choose some
    constants $\eps_0>0$ and $\lambda_0<1/2$. By Theorem~\ref{win-win} and
    Theorem~\ref{small-lambda}, we can solve Subset Sum in $\Os(2^{(0.5 -
    \delta')n})$ time for some $\delta'(\eps_0,\delta_0)$. Hence, from now on we assume
    that $\eps < \eps_0$ and $\lambda > \lambda_0$.

    We use the construction from Lemma~\ref{main-lem}. This Lemma,
    gives us the algorithm that constructs $4$ families of sets:
    $\Ll_1,\Ll_2,\Rr_1,\Rr_2 \subseteq 2^{[n]}$, with the following properties
    (see Lemma~\ref{lem:detect} and Lemma~\ref{main-lem}):
    \begin{itemize}
        \item If an answer to Subset Sum is positive, then with $\Oms(1)$ probability, there exist $S_1 \in \Ll_1, S_2 \in \Ll_2,
            S_3 \in \Rr_2, S_4 \in \Rr_1$, that are pairwise disjoint and $w(S_1\cup S_2
            \cup S_3 \cup S_4 ) = t$ (otherwise there is no such quadruple).
        \item For every $S_1 \in \Ll_1$, $S_2 \in \Ll_2$, $S_3 \in \Rr_2$, $S_4
            \in \Rr_1$ we have that $S_1 \cap S_3 = \emptyset$, $S_2 \cap S_4 =
            \emptyset$ and $S_1 \cap S_4 = \emptyset$.
        \item The expected size of these lists is bounded by $\Os( 
            2^{(1/2 - \mu(3/2 + \lambda - h(1/4)))n+ \mu\rho n} +
        2^{\mu n(2h(1/4) - \lambda) + \mu \rho n})$, for some $\rho =
        \rho(\lambda,\eps)$ that goes to $0$ when $\eps \rightarrow 0$ and
        $\lambda \rightarrow 1/2$ (see Remark~\ref{remark}).
    \end{itemize}
    Define constant
    $\kappa(\eps_0,\lambda_0,\mu)$ that goes to $0$ when $\eps_0 \rightarrow 0$
    and $\lambda_0 \rightarrow 1/2$ such that the expected size of the lists is:
    \begin{displaymath}
        \max \left\{ \Ex{|\Ll_1|}, \Ex{|\Ll_2|}, \Ex{|\Rr_1|}, \Ex{|\Rr_2|}
        \right\} \le \Os( 
                2^{(1/2 - \mu(2 - h(1/4)))n+ \kappa n} +
            2^{\mu n(2h(1/4) - 1/2) + \kappa n})      
            .
    \end{displaymath}
    Next we select $\mu := 1/(3 + 2h(1/4)) \approx 0.2164$ to minimize the
    expected size of $\Ll_1,\Ll_2,\Rr_2,\Rr_1$. We have that:
    \begin{displaymath}
        \max \left\{ \Ex{|\Ll_1|}, \Ex{|\Ll_2|}, \Ex{|\Rr_1|},  \Ex{|\Rr_2|}
        \right\} \le \Os(2^{0.2428432n + \kappa n})
        .
    \end{displaymath}
    Now, we proceed with the reduction to Exact Node Weighted $P_4$. First, we
    construct a graph.
    Let $M := 100 \cdot w([n])$ be sufficiently large integer. For every
    set $A \in \Ll_1$ create a vertex $v^1_A$ of weight $w(v^1_A) = M +
    w(A)$, for every set $B \in \Ll_2$ create a vertex $v^2_B$ of weight
    $2M + w(B)$, for every set $C \in \Rr_2$ create a vertex $v^3_C$ of
    weight $4M + w(C)$. Finally, for every set $D \in \Rr_1$ create a vertex
    $v^4_D$ of weight $-7M - t + w(D)$.

    Next for every $i \in \{1,2,3\}$ add an edge
    between vertices $v^i_X$ and $v^{i+1}_Y$ iff $X \cap Y = \emptyset$. This
    concludes the construction. At the end we run our
    hypothetical oracle to an algorithm for Exact Node Weighted $P_4$ and return
    true if the oracle detects a simple path of $4$ vertices with total weight
    $0$. This concludes description of the reduction.

    Now we analyse the correctness. If there exist $S_1 \in \Ll_1, S_2 \in \Ll_2,
    S_3 \in \Rr_2$ and $S_4 \in \Rr_1$ that are disjoint and sum to $t$, then vertices
    $v^1_{S_1},\; v^2_{S_2},\; v^3_{S_3}, \;v^4_{S_4}$ form a path and their
    sum is equal to $0$.
    For the other direction, suppose there exist $4$ vertices that
    form a path and their sum is equal to $0$. Because their sum is
    equal to $0$ and integer $M$ is larger than the rest of the weight, these
    vertices from $4$ distinct groups, i.e. vertices $v^1_{A},v^2_{B},v^3_{C},v^4_{D}$ for some sets $A,B,C,D$. Moreover vertices
    $v^1_A,v^2_B$ have to be connected (since vertices in group 1 are connected
    only to the vertices in group 2), hence $A \cap B = \emptyset$.
    Analogously, it can be checked that the rest of
    the sets $B,C,D$ are disjoint. Observe that $w(v^1_A) + w(v^2_B) + w(v^3_C)
    +w(v^4_D) = 0$ hence $w(A) + w(B) + w(C) + w(D) = t$. By the correctness of
    construction of $\Ll_1,\Ll_2,\Rr_2,\Rr_1$ we conclude that the answer to the
    Subset Sum instance is positive.

    Finally we analyse the runtime of our reduction.  The number of vertices is
    clearly $|V| = \Os(|\Ll_1| + |\Ll_2| + |\Rr_2| + |\Rr_1|) \le
    \Os(2^{0.2428432n + \kappa n})$ The time needed to construct this graph is
    $|E| = \Os(|\Ll_1||\Ll_2| + |\Ll_2||\Rr_2| + |\Ll_1||\Ll_2|) \le
    \Os(2^{0.48569n + 2\kappa n})$ (this is also an upper bound on number of
    edges).  

    Hence if we would have an algorithm that solves Exact Node Weighted $P_4$ in
    time $\Oh(|V|^{2.05894})$, then Subset Sum could be solved in randomized time
    $\Os(2^{2.05894(0.2428432n + \kappa n)}) \le \Os(2^{(0.4999995 + 2.06\kappa)n})$. Note
    that $\kappa(\eps_0,\lambda_0)$ is some constant that can be selected to be
    arbitrarily close to $0$.
\end{proof}

\section{Orthogonal Vectors via Representative Sets}
\label{sec:ov}

In this section we present and discuss our algorithm for Orthogonal Vectors.
As discussed in the introduction it should be noted that the proof strategy is
similar to the one from~\cite{fomin-jacm2016} (which is heavily inspired on
Bollob\'as's Theorem~\cite{bollobas1965generalized}), but we obtain improvements
that are crucial for the main result of this paper. We compare our methods with
existing literature at the end of this section.

\begin{theorem}[OV-algorithm, Generalization of Theorem~\ref{ov-simplified}]
	\label{ov-lemma}
    For any $\lambda \in [0.4,0.5]$ and $\sigma \in [0.4,0.6]$, there is a Monte-Carlo algorithm that is
    given $\mathcal{A} \subseteq \binom{d}{\sigma \lambda d}$ and $\mathcal{B}
    \subseteq \binom{d}{(1-\sigma)\lambda d}$, detects if there exist $A \in
    \mathcal{A}$ and $B \in \mathcal{B}$ with $A \cap B = \emptyset$ in time
    \begin{displaymath}
        \Ot\left( 
            \left( |\mathcal{A}| + |\mathcal{B}| \right)2^{d(1/2 + \lambda - h(\lambda/2))}
        \right)
    \end{displaymath}
    and space $\Ot(|\mathcal{A}| + |\mathcal{B}| + 2^d)$.
\end{theorem}

We can assume that
$\lambda \le 0.5$ by a subset complementation trick. The bound
$\sigma \in [0.4,0.6]$ is an artifact of technical methods we used in the proof
of Lemma~\ref{ovtime-inequality}.  
In the proof of this lemma the parameters $\lambda$ and $\sigma$ lost their meaning from Section~\ref{main-proof}. Hence, to simplify, we let $p := \sigma \lambda n$ and $q :=
(1-\sigma)\lambda n$, and let $\mathcal{A} \subseteq
\binom{[d]}{p}$ and $\mathcal{B} \subseteq \binom{[d]}{q}$.
We use the following standard definitions from communication complexity (see for example~\cite{rao-yehudayoff-commbook}):
\begin{definition}[$(p,q,d)$-Disjointness Matrix]
	For integers $p,q,d$ the Disjointness matrix $\disj_{p,q,d}$ has its rows indexed by $\binom{[d]}{p}$ and its columns indexed by $\binom{[d]}{q}$. For $A \in \binom{[d]}{p}$ and $B \in \binom{[d]}{q}$ we define
	\[
	 \disj_{p,q,d}[A,B] =
	 \begin{cases}
	 	 1 & \text{if } A \cap B = \emptyset,\\ 
	 	 0 & \text{otherwise. } 
	 \end{cases}
	 \] 
\end{definition}

\begin{definition}[Monochromatic Rectangle, $1$-Cover]
	A \emph{monochromatic rectangle} of a matrix $M$ is subset $X$ of rows and
    subset $Y$ of the columns such that $M[i,j]=M[i',j']$ for every $i,i'\in X$
    and $j, j'\in Y$. A family of monochromatic rectangles
    $\mathcal{M} = (X_1,Y_1),\ldots,(X_z,Y_z)$ is called a \emph{$1$-cover} if for every $i,j$
    such that $M[i,j]=1$, there exists $k \in [z]$, such that $i \in X_k$ and $j
    \in Y_k$.
\end{definition}

A natural goal in the field of communication complexity is to find `good'
$1$-covers. The natural parameter that quantifies such `goodness' is $z$ (intuitively the smaller $z$ the
better a $1$-cover we have). The parameter $z$ is sometimes called the Boolean rank\footnote{The name `Boolean rank' is used because a $1$-cover of $M$ with $z$ rectangles is equivalent to a factorization $M=L\cdot R$ over the Boolean semi-ring of rank $z$.} and it is known to be equal to
$2^{\mathtt{nc}(M)}$ where $\mathtt{nc}(M)$ is the `non-deterministic
communication complexity' of $M$ (see e.g.~\cite{rao-yehudayoff-commbook}).

Such $1$-covers of the Disjointness matrix can be used in algorithms for the Orthogonal Vectors problem: An orthogonal pair is a $1$ in the submatrix of the Disjointness induced by the rows and columns from the families $\Aa$ and $\Bb$, and we can search for such a $1$ via searching for the associated monochromatic rectangle that covers it (see Lemma~\ref{ov-by-sparsity} for a related approach).
For the case that $p=q$, it is well known that $\disj_{p,q,d}$ admits a
$1$-cover with $\Oh(2^{2p}p \ln d)$ rectangles~\cite[Claim 1.37]{rao-yehudayoff-commbook}.
When applied na\"ively, this $1$-cover would imply an $\Ot((|\Aa|+|\Bb|)2^{d/2})$
time algorithm for the setting of Theorem~\ref{ov-simplified} with $p=q=d/4$. 

In order to get a faster algorithm we introduce the following new parameter of a $1$-cover:

\begin{definition}[Sparsity]
    The \emph{sparsity} of a $1$-cover $\mathcal{M} =
    (X_1,Y_1),\ldots,(X_z,Y_z)$ of an $n \times m$ matrix is defined as $\sum_i
    |X_i|/n + \sum_i |Y_i|/m$.
\end{definition}

A $1$-cover of sparsity $\Psi$ of a matrix can be understood as a factorization
of $M=L\cdot R$ over the Boolean semi-ring such that the average number of $1$'s in a
row $L$ plus the average number of $1$'s in a column of $R$ is at most $\Psi$.
Our notion of sparsity is related to the degree of the data structure called
\emph{$n$-$p$-$q$-separating collection}~\cite{fomin-jacm2016}.
For a further discussion about sparse factorizations see~\cite[Section 5.1]{nederlof-ranksurvey})

We present the algorithmic usefulness of the notion of the sparsity of $1$-cover with the following
statement.

\begin{lemma}[Orthogonal Vectors Parameterized by the Sparsity]
    \label{ov-by-sparsity}

    For any constant integer $c$ and integers $p,q,d$
    such that $c$ divides $p,q,d$, there is an algorithm that takes as an input
    a $1$-cover $\mathcal{M}$ of $\disj_{p/c,q/c,d/c}$ of sparsity
    $\Psi$ and two set families $\Aa \subseteq \binom{[d]}{p}$, $\Bb \subseteq
    \binom{[d]}{q}$ with the following properties:
    It outputs a pair $A \in \mathcal{A}$ and $B \in \mathcal{B}$ such
    that $A \cap B = \emptyset$ with constant non-zero probability if such a
    pair exists. Moreover, it uses $\Ot((|\Aa|+|\Bb|)\Psi^c+2^{2d/c})$ time and
    $\Ot(2^{2d/c}+z^c)$ space, where $z$ is the number of rectangles of $\mathcal{M}$.
\end{lemma}

\begin{algorithm}
	\SetKwInOut{Input}{Input}
	\SetKwInOut{Output}{Output}
    \Input{$\mathcal{A} \subseteq \binom{d}{p}, \mathcal{B} \subseteq
    \binom{d}{q}$ and $1$-cover $\mathcal{M} = (X_1,Y_1),\ldots,(X_z,Y_z)$.}
    \Output{Exist $A \in \mathcal{A}$ and $B \in \mathcal{B}$ such that $A\cap B = \emptyset$?}
    \DontPrintSemicolon
    Randomly partition $[d] = [U_1] \uplus \ldots \uplus [U_c]$  \\
    For every $Q \in \binom{[U_i]}{p/c}$ construct $L_i(Q) := \{ j \in [z] \; : \; Q \in X_j\}$\tcp*{Use $\Os(2^{2d/c})$ space}
    For every $Q \in \binom{[U_i]}{q/c}$ construct $R_i(Q) := \{ j \in [z] \; : \; Q \in Y_j\}$\tcp*{Use $\Os(2^{2d/c})$ space}

    Initialize $T[i_1,\ldots,i_c] = \mathtt{False}$ for every $i_1,\ldots,i_c
    \in [z]$ \tcp*{Use $\Ot(z^c)$ space}
    \ForEach{$A \in \mathcal{A}$}{
        \If{$\prod_{i=1}^c |L_i(A \cap U_i)| \le \Ot(|\Psi|^c)$}{
            \ForEach{$(i_1,\ldots,i_c) \in L_1(A \cap U_1) \times \ldots \times
            L_c(A \cap U_c)$}{
            Set $T[i_1,\ldots,i_c] = \mathtt{True}$
            }
        }
    }
    \ForEach{$B \in \mathcal{B}$}{
        \If{$\prod_{i=1}^c |R_i(B \cap U_i)| \le \Ot(|\Psi|^c)$}{
            \ForEach{$(i_1,\ldots,i_c) \in R_1(B \cap U_1) \times \ldots \times
            R_c(B \cap U_c)$}{
                \If{$T[i_1,\ldots,i_c] = \mathtt{True}$}{
                    \Return $\mathtt{True}$
                }
            }
        }
    }
    \Return $\mathtt{False}$.
    \caption{Pseudocode of Lemma~\ref{ov-by-sparsity}}
    \label{alg:ov-by-sparsity}
\end{algorithm}
\begin{proof}
	Denote the $1$-cover to be $\mathcal{M} = (X_1,Y_1),\ldots,(X_z,Y_z)$.
    Observe, that if $A \cap B = \emptyset$ then it suffices to find $\ell \in
    [z]$ such that $A \in X_\ell$ and $B \in Y_\ell$ since $\mathcal{M}$ forms a
    $1$-cover. In the bird's eye view, the algorithm will find such an $\ell$.
    We need to make sure that the space usage of our algorithm is low.  We
    will use parameter $c$ to achieve that (it is instructive for a reader to
    assume $c = 1$).  Algorithm~\ref{alg:ov-by-sparsity} presents an
    overview of the proof.

    First, randomly partition $[d]$ into blocks $U_1,\ldots,U_c$ with $|U_d|=d/c$.
    By Lemma~\ref{lem:conc}, if we repeat the algorithm $d^{\Oh(c)}$ times with probably
    at least $1/d^{O(c)}$ this partition is good, i.e., for some orthogonal pair
    $A,B$ it holds that $|A \cap U_i|=p/c$ and $|B \cap U_i|=q/c$.

    Next, we map the given factorization $X_1,\ldots,X_z,Y_1,\ldots,Y_z$ of
    $\disj_{p/c,q/c,d/c}$, to the set $U_i$ by unifying $U_i$ with $[d]$ a uniformly
    random permutation.

    Now we present a processing step of the algorithm. For every $i \in [c]$
    we create and store two lists $L$ and $R$. The purpose of these lists is
    to give every element in $A$ and $B$ fast access to corresponding rectangles
    from the $1$-cover that contain it (i.e., given $A$ we need to find all $X_i$, such that $A \in
    X_i$ in $\Ot(z)$ time). Specifically, for every $i \in [c]$ construct:
    \begin{displaymath}
        \text{For every set } Q \in \binom{[d/c]}{p/c} \text{ construct the list } L_i(Q) := \{ j \in [z] \; : \; Q \in X_j \}
        .
    \end{displaymath}
    And similarly for all $i \in [c]$:
    \begin{displaymath}
        \text{For every set } Q \in \binom{[d/c]}{q/c} \text{ construct the list } R_i(Q) := \{ j \in [z] \; : \; B \in Y_j \}
        .
    \end{displaymath}
    
    Because $\binom{d/c}{p/c} \le 2^{d/c}$
    we can construct and store all $L_i(Q)$ and $R_i(Q)$ in
    $\Ot(2^{2d/c} + 2^{d/c}z)$ time and space. Additionally, initialize a table $T[i_1,\ldots,i_c] := \mathtt{False}$ for every
    $i_1,\ldots,i_c \in [c]$. This table will store which sets $X_i,Y_i$ have
    been seen by elements in $\mathcal{A}$. 
    Observe that so far we did not look at the
    input $\mathcal{A}$ and $\mathcal{B}$; we just preprocessed the $1$-cover, so
    the next steps can be computed efficiently.

    Now iterate over every element $A \in \mathcal{A}$ and check if we
    can afford to process it, i.e., if 
    $|L_1(A\cap U_1)| \cdots |L_c(A \cap U_c)| > (4c\Psi)^c$ we simply ignore it
    (later we will prove that for a disjoint pair $A$ and $B$ this situation
    happens with low probability). If indeed we can afford it, then we mark it in
    table $T$: For every $(i_1,\ldots,i_c) \in L_1(A \cap U_1)\times \ldots
    \times L_c(A \cap U_c)$ we mark $T(i_1,\ldots,i_c)$ to be $\mathtt{True}$. Clearly
    this step takes $\Ot(|\mathcal{A}| \Psi^c)$ time.

	Next, we treat $\mathcal{B}$ in a similar way: We iterate over every element $B \in \mathcal{B}$ and check if 
    $|R_1(B\cap U_1)| \cdots |R_c(B \cap U_c)| \leq (4c\Psi)^c$. If so, we iterate over every $(i_1,\ldots,i_c) \in R_1(B \cap U_1)\times \ldots
    \times R_c(B \cap U_c)$ and check if $T(i_1,\ldots,i_c) = \mathtt{True}$. If this
    happens, then it means there exists $A \in \mathcal{A}$ that is
    orthogonal to the current $B$ and we can return $\mathtt{True}$. If this never
    happens, we return $\mathtt{False}$. Clearly, the total running time of the algorithm is
    $\Ot\left( (|\mathcal{A}| + |\mathcal{B}|) \Psi^c\right)$ and extra amount of working memory is
    $\Ot(2^{2d/c} + z^c)$. Hence we focus on correctness.
	
    Note that if $\mathtt{True}$ is returned, indeed there must exist
    disjoint $A \in \Aa$ and $B \in \Bb$ because $\mathcal{M}$ is $1$-cover.
    For the other direction, suppose that there exist orthogonal $A \in \Aa$ and $B \in \Bb$.
    As mentioned this implies by Lemma~\ref{lem:conc} that with $1/d^{c}$ we have that for each $i$ it holds that $|A \cap U_i|=p/c$ and $|B \cap
    U_i|=p/c$.  Because we unified $[d]$ with $U_i$ with a random permutation,
    $\mathbb{E}[|L_i(A \cap U_i)|],\mathbb{E}[|R_i(A \cap U_i)|]= \Psi$, and
    by Markov's inequality and a union bound there will be no $i$ with $|L_i(A
    \cap U_i)|+|R_i(B \cap U_i)| \geq 4c\Psi$, and therefore $|L_1(A\cap U_1)|
    \cdots |L_c(A \cap U_c)| \le (4c\Psi)^{c}$ and $|R_1(B \cap
    U_1)| \cdots |R_c(B \cap U_c)| \le (4c\Psi)^{c}$. If this happens,
    the orthogonal pair will be detected since $(X_1,Y_1),\ldots,(X_z,Y_z)$ is a
    $1$-cover.
\end{proof}

\begin{lemma}[Construction of $1$-cover with small sparsity]
    \label{lem:factor}

    Let $p,q$ and $d$ be integers such that $p \leq q$ and $p+q\leq d/2$. There is a randomized
    algorithm that in $\Oh(2^{d})$ time and space, constructs $X_1, \ldots,X_z
    \subseteq \binom{[d]}{p}$ and $Y_1, \ldots,Y_z \subseteq \binom{[d]}{p}$,
    where $z$ is at most $2^d$.

    All pairs of sets $(X_1,Y_1),\ldots,(X_z,Y_z)$ form monochromatic rectangles in $\disj_{p,q,d}$
    and with probability at least $3/4$, it holds that $(X_1,Y_1)\ldots,(X_z,Y_z)$ is a $1$-cover of $\disj_{p,q,d}$ with sparsity
	\[
        d^{\Oh(1)} \cdot 2^{d/2 + p+q - d\cdot h\left(\frac{p+q}{2d}\right)}.
	\]
\end{lemma}
\begin{proof}
    Let $l=p+q$ and let $A \in \mathcal{A}, B \in \mathcal{B}$ be an orthogonal
    pair. Let $x$ be some parameter that we will determine later (think about $x \approx d/2$). Note that
	\[
        \left|\left\{ S \in \binom{[d]}{x} :\;  A \subseteq S \text{ and } S \cap B = \emptyset \right\}\right| = \binom{d-l}{x-p}.
	\]
    Let $\Ss := \{S_1,\ldots,S_z\} \subseteq \binom{[d]}{x}$ be obtained by
    including each set from $\binom{[d]}{x}$ with probability
    $2d\binom{d-l}{x-p}^{-1}$ (assuming $x > p +\Omega(1)$, this probability is indeed in the interval $[0,1]$).
   
    Thus, if $A$ and $B$ are disjoint sets, with good
    probability there is a certificate set $S \in \Ss$, such that $A \subseteq
    S$ and $S \cap B = \emptyset$. More formally:
    \begin{equation}
        \label{eq:prob-cover}
        \prob{\not\exists S \in \Ss: A \subseteq S \text{ and } S \cap B=\emptyset \; \mid \;
        A \cap B = \emptyset}{} = \left(1-2d\binom{d-l}{x-p}^{-1}\right)^{\binom{d-l}{x-p}} \leq \exp(-2d),
    \end{equation}
    (where the last inequality is due to the standard inequality $1+\alpha\leq \exp(\alpha)$).
    Now we define a $1$-cover based on the family $\Ss$:
    \begin{displaymath}
        \text{For every } i \in [z]: \;\;\;\; X_i := \binom{S_i}{p} \; \text{ and } \; Y_i:= \binom{[d]\setminus S_i}{q}
        .
    \end{displaymath}
    First let us prove that with good probability $X_1,Y_1\ldots,X_z,Y_z$ is
    $1$-cover. There are at most $3^d$ disjoint pairs $A,B$. Hence by
    Equation~\ref{eq:prob-cover} and the union bound on all disjoint pairs $A,B$, we have that
    $(X_1,Y_1),\ldots,(X_z,Y_z)$ is a $1$-cover with probability at least
    $3/4$. 

    Next, we bound the sparsity of $(X_1,Y_1),\ldots,(X_z,Y_z)$. By Markov's
    inequality, $z \leq 4d\binom{d}{x}\binom{d-l}{x-p}^{-1}$ with probability
    at least $1/2$. Hence with probability at least $1/2$ our $1$-cover
    has sparsity at most:
    \begin{equation}
        \begin{aligned}
              &4d\binom{d}{x}\binom{d-l}{x-p}^{-1} \bigg(|X_i| / \binom{d}{p} + |Y_i| /   \binom{d}{q}\bigg) =\\
            &4d\binom{d}{x}\binom{d-l}{x-p}^{-1}\left(\binom{x}{p}\binom{d}{p}^{-1}+\binom{d-x}{q}\binom{d}{q}^{-1}\right) =\\
            &4d\bigg(\binom{d-p}{x-p}+\binom{d-q}{x}\bigg)\binom{d-l}{x-p}^{-1},
        \end{aligned}
        \label{eq:sparsitybound}
    \end{equation}
    where the second equality follows from using $\binom{a}{b+c}\binom{b+c}{c}=\binom{a}{b,c}=\binom{a}{c}\binom{a-c}{b}$ twice.

    Next, we use  Lemma~\ref{ovtime-inequality} (see
    Appendix~\ref{sec:ineq-ov}) with: $d=n$, $p=\sigma\lambda n$,
    $q=(1-\sigma)\lambda n$ and $p+q = \lambda n$. Note that we assumed that $\sigma \in
    [0.4,0.6]$ and $\lambda \in [0.4,0.5]$ hence conditions for
    Lemma~\ref{ovtime-inequality} are satisfied. We obtain that for the choice
    of $x:=d(1/2 + (\sigma - 1/2)(\log_2(3)/2) + (1/2 - \sigma)(1/2 - \lambda))$
    expression \eqref{eq:sparsitybound} is bounded from above with
	\[
        d^{\Oh(1)} \cdot 2^{d/2 + p+q - d\cdot h(\frac{p+q}{2d})},
	\]
	as required.
\end{proof}

Now the main statement of this section follows by a straightforward combination of the previous lemmas:

\begin{proof}[Proof of Theorem~\ref{ov-lemma}]
    Let $p := \sigma \lambda d$ and $q := (1-\sigma) \lambda d$.  Set $c = 20$ and assume that integers $p,q,d$ are multiples of $c$ (by padding the
    instance if needed).

    Next, use Lemma~\ref{lem:factor} with $d/c$, $p/c$ and $q/c$ to
    construct a $1$-cover $\mathcal{M}$ of sparsity
	\[
        \Psi = d^{\Oh(1)} \cdot 2^{\frac{d/2 + p+q - d h((p+q)/(2d))}{c}},
	\]
    with good probability.	Subsequently, apply Lemma~\ref{ov-by-sparsity} with this
    $1$-cover $\mathcal{M}$ to detect a disjoint pair $A \in \Aa$ and $B \in \Bb$ with
    constant probability. Note that the runtime is:
	\begin{align*}
		\Otilde\left(\left( |\mathcal{A}| + |\mathcal{B}| \right)(4c\Psi)^c + 2^{2d/c}\right) 
		=\Otilde\left(\left( |\mathcal{A}| + |\mathcal{B}| \right)2^{d/2 + p+q - d h((p+q)/2d)}\right)
        .
    \end{align*}
    Hence, the running time is $\Ot\left( 
		\left( |\mathcal{A}| + |\mathcal{B}| \right)2^{d(1/2 + \lambda -
    h(\lambda/2))}\right)$. The main bottleneck in the space usage comes from the $z^c$ factor in Lemma~\ref{ov-by-sparsity} which gives the $2^d$ factor.
\end{proof}

\paragraph{Lower bound on sparsity}

One might be tempted to try to get even better bounds on the sparsity of the disjointness matrix. 
Here we show that the sparsity bound from Lemma~\ref{lem:factor} is
essentially optimal with a fairly straightforward counting argument. It means
that new techniques would have to be developed to improve an
algorithm for Orthogonal Vectors in the worst case $\sigma = 1/2$ and $\lambda =
1/2$, and in consequence improve the meet-in-middle algorithm for Subset Sum.

\begin{theorem}
	Any $1$-cover of $\disj_{d/4,d/4,d/2}$ has sparsity at least
    $\Oms\left(2^{d}/\binom{d}{d/4}\right)$.
\end{theorem}
\begin{proof}
    Let $(X_1,Y_1),\ldots,(X_z,Y_z)$ be a $1$-cover of $\disj_{d/4,d/4,d/2}$.
    Next, we define 
	\[
        L_i: = \bigcup_{ A \in X_i} A \qquad \text{ and }\qquad R_i := \bigcup_{B \in Y_i} B.
	\]
    We say an index $i \in [z]$ is \emph{left-heavy} if $|L_i| > d/2$ and \emph{right-heavy} if $|R_i| > d/2$.
	Note that $i$ cannot both be left-heavy and right-heavy since otherwise there exist $A \in X_i$ and $B \in Y_i$ that overlap, contradicting that $X_i,Y_i$ is a monochromatic rectangle.
	
	By swapping $X_i$ and $Y_i$ we can assume without loss of generality that
	\[
        \sum_{\substack{i\in[z] \\ i \text{ left-heavy}}} |X_i||Y_i| \geq
        \sum_{\substack{i \in [z] \\ i \text{ right-heavy}}} |X_i||Y_i|
        .
	\]
	Since every disjoint pair of sets $A,B \subseteq [d]$ with $|A|=|B|=d/4$ must be in at least one rectangle, we have the lower bound
	\[
	\binom{d}{d/4,d/4} \leq \sum_{i=1}^z |X_i||Y_i| \leq 2\sum_{i \text{ left-heavy}} |X_i||Y_i|\leq 2\sum_{i \text{ left-heavy}} |X_i|\binom{d/2}{d/4},
	\]
	where the last inequality holds since $|A_i| \leq d/2$ implies that $|Y_i| \leq \binom{d/2}{d/4}$.
	Thus $\sum_{i=1}^z |X_i| \geq 2^{d}/d^{\Oh(1)}$, and the theorem follows
    because rows and columns are $\binom{d}{d/4}$.
\end{proof}

\paragraph{Relation of Techniques in this section with existing methods}
The idea for constructing the $1$-cover is relatively standard in communication
complexity (see e.g., the aforementioned~\cite[Claim
1.37]{rao-yehudayoff-commbook}). It was also used in some proofs of Bollob\'as's Theorem~\cite{bollobas1965generalized}. The idea of randomly partitioning the universe to get a structured $1$-cover is very similar to the derandomization of the color-coding approach from~\cite{AlonYZ95}. 

Both ideas were also used by~\cite{fomin-jacm2016}. They also
start with a probabilistic construction (c.f., \cite[Lemma 4.5]{fomin-jacm2016}) on a small universe that is
repeatedly applied, and use it to set up a data structure of `$n$-$p$-$q$-separating
collections' that is similar to our lists.\footnote{Additionally they
derandomize their construction by using brute-force to find the probabilistic
construction and use splitters to derandomize the step of splitting the universe
into $c$ blocks.} The small but crucial difference, however is that (in our
language) they obtain a monochromatic rectangle by sampling a random set $S
\subseteq [d]$ (in contrast to our random sampling $S \in \binom{[d]}{d/4}$ in
the case $p=q=d/4$), and in the case $p=q=d/4$ this would lead to sparsity
$2^{3d/4}/2^{d/4} \gg 2^d/\binom{d}{d/4}$.

\paragraph{Acknowledgements.} The first author would like to thank Per Austrin,
Nikhil Bansal, Petteri Kaski, Mikko Koivisto for several inspiring discussions
about reductions from Subset Sum to Orthogonal Vectors. The second author would
like to thank Marcin Mucha and Jakub Pawlewicz for useful discussions.

\bibliographystyle{alpha}
\bibliography{bib}

\begin{appendices}
\section{Omitted proofs}
\label{proof-of-lem:sumsetenum}

\subsection{The Approach of Schroeppel and Shamir}
\label{schroeppel-shamimir-algorithm}
In this section we recall the Approach from \cite{schroeppel} in such a way that we can easily reuse parts of it. Their crucial insight is formalized in Lemma~\ref{lem:sumsetenum}, which we first recall for convenience: 
\incdatastructure*
\begin{algorithm}
\SetKwInOut{preprocessing}{Preprocessing $\inc(A,B)$}
\SetKwInOut{next}{Operation \inc.\nextel}
\preprocessing{}
Sort $A = \{a_1,\ldots,a_k\},B = \{b_1,\ldots,b_l\}$\\
Initialize priority queue $Q$\\
For every $b_j \in B$ add $(a_1,b_j)$ to $Q$ with priority $a_1 + b_j$\\
\next{}
\setcounter{AlgoLine}{0}
\LineIf {$Q$ is empty} { \Return $\mathtt{EMPTY}$}\\
Let $w$ be the lowest priority in the queue $Q$\\
Initialize $P^l_w := \emptyset$\\
\While {the lowest priority in $Q$ is $w$} {
    Let $(a_i,b_j)$ be the element with lowest priority in queue $Q$ and remove it\\
    \LineIf {$i \le k$} {add $(a_{i+1},b_j)$ with priority $a_{i+1}+b_j$ to $Q$ }\\
    $P^l_w = P^l_w \cup \{ (a_i,b_j) \}$ \tcp*{$a_i+b_j = w$}
}

\LineIf {$P^l_w = \emptyset$}{
        \Return \inc.\nextel \tcp*{Seek next if $w \notin C$}
}

\Return $(P^l_w,w)$.
\caption{Pseudocode of Lemma~\ref{lem:sumsetenum}}
\label{alg:inc}
\end{algorithm}
\begin{proof}
    For the overview of the proof see Algorithm~\ref{alg:inc}. During the
    preprocessing step we sort sets $A$ and $B$ in increasing order $A = \{a_1 \le
    a_2\le\ldots \le a_k\}$ and $B = \{b_1\le b_2 \le \ldots \le b_l\}$. Next we
    initialize priority queue $Q$ and for every $b_j \in B$ 
    we add a tuple $(a_1,b_j)$ with priority $a_1 + b_j$. The
    preprocessing clearly takes $\Ot(|A| + |B|)$ time and space.

    Now we explain the implementation of operation \inc.\nextel. We let $w$ be
    the priority of the element with the lowest priority in our queue. We go over all
    $(a_i,b_j)$ in the priority queue that have the priority $w$ (and therefore
    $a_i + b_j = w$) and add them to $P^l_w$. Namely, we remove every element
    $(a_i,b_j)$ with $a_i+b_j = w$ from the queue and replace it with
    $(a_{i+1},b_j)$. For correctness, note that every pair $(a_i,b_j) \in A
    \times B$ will eventually be added and removed from the queue. Moreover the
    priority queue outputs elements in the increasing order.

    For the space complexity, observe that at any moment for every $b \in B$
    there exists at most one $a \in A$ such that $(a,b) \in Q$. Hence at any
    moment the size of the priority queue is $\Ot(|B|)$ space.

    For the running time, observe that every pair $(a,b) \in A \times B$ will
    be added and removed from $Q$ exactly once. Hence the total running
    time of all calls to $\inc.\nextel$ is $\Ot(|A||B|)$.
\end{proof}

The datastructure of Lemma~\ref{lem:sumsetenum} can be used for an efficient $4$-SUM algorithm.

\begin{lemma}\label{ref:4sum}
	An instance $A,B,C,D,t$ of $4$-SUM with $|A|=|B|=|C|=|D|=N$ can be solved using $\Ot(N^2)$ time and $\Ot(N)$ space.
\end{lemma}
\begin{proof}
	The idea is to simulate a standard linear search routine on sets $A+B$ and $C+D$.
	Use Lemma~\ref{lem:sumsetenum} to enumerate $A+B$ in increasing order and $C+D$ in
    decreasing order. In any iteration with current items $x \in A+B$ and $y \in C + D$, compare $x+y$ with $t$. If $x+y=t$ output YES. Otherwise, if $x+y < t$, query the next (larger) element of $A+B$, and if $x+y> t$ query the next (smaller) element of $C+D$ and iterate. If this terminates output NO. It is easy to see that this is always correct and runs in the required time and space bounds.
\end{proof}
\begin{theorem}[\cite{schroeppel}]
    Given a set $S$ of $n$ positive integers and a target $t \in \nat$.
    In $\Os(2^{n/2})$ time and $\Os(2^{n/4})$ space we can determine if there
    exists $S' \subseteq S$, such that $w(S') = t$.
\end{theorem}
\begin{proof}
    First, arbitrarily partition the input weights into sets $A_1, A_2,A_3,A_4$  with $|A_1|=|A_2|=|A_3|=|A_4| = n/4$. Next enumerate
    and store for every $i \in \{1,\ldots,4\}$ the sets:

    \begin{displaymath}
        \mathcal{A}_i := \{ w(S) \; | \; S \subseteq A_i \}
        .
    \end{displaymath}

    Observe, that $|\mathcal{A}_i| = 2^{n/4}$. We can construct and store
    $\mathcal{A}_i$ for every $i \in \{1,\ldots,4\}$ in $\Os(2^{n/4})$ time and
    space. Now, we solve the 4-SUM instance with sets
    $\mathcal{A}_1,\ldots,\mathcal{A}_4$ and target $t$ using the algorithm from Lemma~\ref{ref:4sum}. Because we can solve an
    instance of 4-SUM of $N$ integers in $\Ot(N^2)$ time and $\Ot(N)$ space and
    an instance size is $N = 2^{n/4}$, the algorithm runs in $\Os(2^{n/2})$ time
    and $\Os(2^{n/4})$ space. For the correctness, assume that $S' \subseteq S$
    with $w(S') = t$. Note that $A_i \cap S' \in \mathcal{A}_i$ for every $i \in
    \{1,\ldots,4\}$. Therefore 4-SUM algorithm answers \emph{yes} if there
    exists  $S' \subseteq S$ with $w(S')=t$. The other direction of correctness is trivial.
\end{proof}

\subsection{Enumerating $\Ll_1,\Ll_2,\Rr_1,\Rr_2$}
\label{enum-appendix}

\begin{lemma}
    \label{enumeration-lemma}
    We can enumerate $\Ll_1,\Rr_1,\Ll_1,\Rr_2$ in
    $\Ot(|\Ll_1|+|\Ll_2|+|\Rr_1|+|\Rr_2|+2^{\mu n})$ time and space.
\end{lemma}
\begin{proof}
    We start with enumerating sets $2^{S_1},\ldots,2^{S_8}$ defined in
    Equation~\ref{eq:conditions}. We can do that in time and space
    $\Os(2^{|L|} + 2^{|M|} + 2^{|R|})$. This is bounded by our claimed runtime,
    because $|L|$ and $|R|$ are bounded by $(1-3\mu + |\beta| \mu )n/2
    < (1-2.988\mu)n/2 < \mu$ (recall that $\mu > 0.21$ and $|\beta| < 0.012$).

    Next, we construct tables modulo $p_L,p_R$, for all $i \in \{1,\ldots,4\}$ and
    $a \in [p_L]$:
    \begin{displaymath}
        T_L^i[a] = \{ X \; | \; X \in 2^{S_i} \text{ and } w(X) \equiv_{p_L} a \}
        .
    \end{displaymath}

    And for all $i \in \{5,\ldots,8\}$ and $a \in [p_R]$:
    \begin{displaymath}
        T_R^i[a] = \{ X \; | \; X \in 2^{S_i} \text{ and } w(X) \equiv_{p_R} a \}
        .
    \end{displaymath}

    Now construct sets $\Ll_1, \Ll_2, \Rr_1,\Rr_2$ with the dynamic programming according
    to Equations~(\ref{eq:l1})-(\ref{eq:r2}). For example, to construct $\Ll_1$ we
    join all sets $X \in T_L^1[a]$ with $Y \in T_L^2[x_L-a]$ for all $a \in
    [p_L]$.

    This can be computed in the $\Os(|\Ll_1| + |\Ll_2| + |\Rr_1| + |\Rr_2| + 2^{\mu n})$
    extra time and space (note that $p_L,p_R = \Os(2^{\mu n})$).
\end{proof}

\subsection{Improved Time-Space Trade-off}

\begin{corollary}
	Let $\mathcal{S}$ be an integer satisfying $\mathcal{S} \leq 2^{0.249999n}$.
    Then any Subset Sum instance on $n$ integers can be solved by a Monte Carlo
    algorithm using $\Os(\mathcal{S})$ space and $\Os(2^{n}/\mathcal{S}^{2.000008})$ time.
\end{corollary}
\begin{proof}
	Let $w_1,\ldots,w_n,t$ be such an instance, and set $b=n-\log_2\mathcal{S}/0.249999$.
	For every subset of $X \subseteq {n-b+1,\ldots,n}$ solve the Subset Sum instance with weights $w_1,\ldots,w_{n-b}$ and target $t-\sum_{i \in X}w_i$ using Theorem~\ref{main-thm}.
	This is clearly a correct Monte Carlo algorithm, and it uses
    $\mathcal{S}=2^{0.249999(n-b)}$ space and time
	$\mathcal{T} = \Os(2^{b+(n-b)/2})$. Thus we have that $\mathcal{T}\mathcal{S}^{0.5/0.249999n}\leq 2^n$. 
\end{proof}

\section{Inequality in the runtime analysis of algorithm for OV}
\label{sec:ineq-ov}

In this section we will prove the bound on the running time of Orthogonal Vectors
algorithm. Intuitively, it means that the hardest case is when $\sigma=1/2$. We
will use the short binomial notation, i.e., $\binom{\alpha n}{\beta n} =
\binom{\alpha}{\beta}_n$. The inequality that we prove is:

\begin{lemma}
    \label{ovtime-inequality}
    For large enough $n$ and $\lambda  \in [0.4,0.5]$ and $\sigma \in [0.4,0.6]$ the following inequality holds:
    \begin{align*}
        \min_x \left\{
             \frac{\binom{1-\lambda\sigma}{x - \lambda\sigma}_n +
             \binom{1-(1-\sigma)\lambda}{x}_n}{\binom{1-\lambda}{x -
             \lambda\sigma}_n}
        \right\} \le
        2^{n(1/2 + \lambda - h(\lambda/2))} n^{\Oh(1)}.
    \end{align*}
\end{lemma}

The strategy behind the proof is to find an $x$ that is a good approximation (up
to a $3$rd order factors) of the equation $\binom{1-\lambda\sigma}{x -
\lambda\sigma}_n = \binom{1-(1-\sigma)\lambda}{x}_n$. We found it with a
computer assistance. Next we plug in the $x$ and use the Taylor expansion up to the
$2$nd order. It will turn out that all the 0th and 2nd order terms cancel out.
Moreover we will prove that $2$nd order terms are negative. Because we use
Taylor expansions, we need an extra assumption about the closeness of
$\lambda,\sigma$ to $1/2$ (recall, that we only need $\lambda \le 0.5$). Observe that $2^{nh(\alpha)}$ is within polynomial factors from
$\binom{n}{\alpha n}$, hence in the proof we decided to skip factors $n^{\Oh(1)}$.

\begin{proof}[Proof of Lemma~\ref{ovtime-inequality}]

    First we do the substitution: $\sigma := 1/2 - \alpha$ and $\lambda := 1/2 -
    \beta$. We have that $\alpha \in [-\frac{1}{10},\frac{1}{10}]$ and $\beta \in [0,\frac{1}{10}]$ by
    the symmetry. The inequality that we need to prove is therefore:

    \begin{align*}
        \min_x \left\{
        \frac{
            \binom{\frac{3}{4} - \alpha \beta + (\alpha+\beta)/2}{x-1/4-\alpha\beta + (\alpha+\beta)/2}_n
            +
            \binom{\frac{3}{4} + \alpha\beta - (\alpha-\beta)/2}{x}_n
        }{
            \binom{1/2+\beta}{x-1/4 - \alpha\beta+(\alpha+\beta)/2}_n
        }
        \right\}
        \le 2^{n(1-\beta - h(1/4-\beta/2))}
    \end{align*}

    Our choice for the minimizer is $x := 1/2 + (\sigma -
    1/2)(\log_2(3)/2) + (0.5-\lambda)(0.5-\sigma) = 1/2 -\alpha (\log_2(3))/2 +
    \alpha\beta$. Moreover define constant $c:=  (\log_2(3) - 1)/2$. Then our
    inequality is:

    \begin{align*}
        \frac{
            \binom{\frac{3}{4} + \alpha/2 + (\beta/2 - \alpha\beta)}{1/4 - c\alpha + \beta/2}_n
            +
            \binom{\frac{3}{4} - \alpha/2 + (\beta/2 + \alpha\beta)}{1/4 + c\alpha + \beta/2}_n
        }{
            \binom{1/2+\beta}{1/4 + \beta/2  - c \alpha}_n
        }
        \le
        2^{n(1-\beta - h(1/4-\beta/2))}
    \end{align*}

    Now we will use the following observation:

    \begin{claim}
        \label{ineq:obs1}
        If $\alpha \in [-\frac{1}{10},\frac{1}{10}],\beta \in (0,\frac{1}{10})$ then 
        \begin{displaymath}
            \binom{1/2+\beta}{1/4 + \beta/2  - c \alpha}_n \ge 2^{n(1/2+\beta - 0.7 \alpha^2)}
        \end{displaymath}
    \end{claim}

    Hence if we multiply by the divisor, our inequality is simplified to:
    \begin{align*}
        \binom{\frac{3}{4} + \alpha/2 + (\beta/2 - \alpha\beta)}{1/4 - c\alpha + \beta/2}_n
            +
            \binom{\frac{3}{4} - \alpha/2 + (\beta/2 + \alpha\beta)}{1/4 + c\alpha + \beta/2}_n
        \le
        2^{n(3/2 - 0.7 \alpha^2 - h(1/4-\beta/2))}
    \end{align*}
    Next, we use the inequality $\binom{n+\eps}{k} \le 2^\eps \binom{n}{k}$ twice
    (for $\eps = \alpha\beta$ and $\eps = -\alpha\beta$) to simplify to:
    \begin{align*}
        2^{\alpha\beta} \left(  
            \binom{\frac{3}{4} + \alpha/2 + \beta/2}{1/4 - c\alpha + \beta/2}_n
            +
        \binom{\frac{3}{4} - \alpha/2 + \beta/2)}{1/4 + c\alpha + \beta/2}_n
        \right)
        \le
        2^{n(3/2 - 0.7 \alpha^2 - h(1/4-\beta/2))}
    \end{align*}
    Next, we use inequality $\binom{n+\eps}{k+\eps} \le 2^\eps \binom{n}{k}$ for
    $\eps = \beta/2$ and simplify it even further:
    \begin{align*}
        2^{\alpha\beta + \beta/2} \left(  
            \binom{\frac{3}{4} + \alpha/2}{1/4 - c\alpha}_n
            +
            \binom{\frac{3}{4} - \alpha/2}{1/4 + c\alpha}_n
        \right)
        \le
        2^{n(3/2 - 0.7 \alpha^2 - h(1/4-\beta/2))}
    \end{align*}
    Next, we use the following:
    \begin{claim}
        \label{ineq:taylorh4}
        For every $\alpha \in [-\frac{1}{10},\frac{1}{10}],\beta \in [0,\frac{1}{10}]$ it holds:
        \begin{displaymath}
            h(1/4 - \beta) \le h(1/4) - \beta /2 - \alpha \beta.
        \end{displaymath}
    \end{claim}
    Using this claim, it remains to show that
    \begin{align*}
        \binom{\frac{3}{4} + \alpha/2}{1/4 - c\alpha}_n
            +
            \binom{\frac{3}{4} - \alpha/2}{1/4 + c\alpha}_n
        \le
        2^{n(3/2 - 0.7 \alpha^2 - h(1/4))}
    \end{align*}
    Finally, we use our last claim:
    \begin{claim}
        \label{ineq:taylorlast}
        For every $\alpha \in [-\frac{1}{10},\frac{1}{10}]$ the following holds:
        \begin{displaymath}
            \binom{\frac{3}{4} - \alpha/2}{1/4 + c\alpha}_n \le \binom{3/4}{1/4}_n
            2^{-0.7 \alpha^2 n}
        \end{displaymath}
    \end{claim}
    Hence our inequality boils down to:
    \begin{align*}
        \binom{3/4}{1/4}_n 2^{- 0.7\alpha^2 n}
        \le
        2^{n(3/2 - 0.7\alpha^2 - h(1/4))}
    \end{align*}
To see that this holds, note that the $\alpha^2$ factors cancel out, and that remaining inequality $\binom{3/4}{1/4}_{n}\binom{1}{1/4}_n \leq 2^{1.5 n}$ is in fact equality because both sides count the number of partitions of $n$ in three blocks of size $n/4$, $n/4$ and $n/2$.
\end{proof}
Now, we will present a proofs of the claims. These are based on the following Taylor
expansions of the entropy function:
\begin{equation}\label{eq:taylorthird}
h(1/3 + x) = h(1/3) + x - \frac{27}{4\ln{8}} x^2 + \frac{27}{8\ln{8}} x^3 - \Oh(x^3)
\end{equation}
We will denote $\kappa_1 := - \frac{27}{4\ln{8}}$ and $\kappa_2 := \frac{27}{8\ln{8}}$.
\begin{equation}\label{eq:taylorfourth}
h(1/4 + x) = h(1/4) + \left(\log_2{3}\right) x - \frac{8}{\ln{8}}x^2 + \Oh(x^3)
\end{equation}
\begin{proof}[Proof of Claim~\ref{ineq:obs1}]

    Recall, that we put $c:=  (\log_2(3) - 1)/2$. First, we use the entropy function and write:
    $$ \binom{1/2+\beta}{1/4 + \beta/2  - c \alpha}_n = 2^{(1/2+\beta)h(1/2 -
    c\alpha/(1/2+\beta))n}.
    $$
    Hence we need to show:
    \begin{displaymath}
        (1/2+\beta) \cdot h\left(1/2 - \frac{c\alpha}{1/2+\beta}\right) \ge 1/2 + \beta - 0.7 \cdot \alpha^2.
    \end{displaymath}
    Next, we use the inequality $h(1/2 - x) \ge 1 - 4x^2$ (Inequality~\eqref{ineq:entropy}) and have:
    \begin{displaymath}
        (1/2+\beta) \cdot h\left(1/2 - \frac{c\alpha}{1/2+\beta}\right) \ge 
        (1/2 + \beta)\left(1-4 \left(\frac{c\alpha}{1/2 + \beta}\right)^2\right) 
        = 1/2 + \beta - \frac{4 c^2 \alpha^2}{1/2+\beta}.
    \end{displaymath}
    Hence, we need to show that:
    \begin{displaymath}
        1/2 + \beta - \frac{4 c^2 \alpha^2}{1/2+\beta} \ge 1/2 + \beta - 0.7 \cdot \alpha^2
    \end{displaymath}
    Which is equivalent to
    \begin{displaymath}
        1/2 + \beta \ge \frac{4 c^2}{0.7}
    \end{displaymath}
    Note that $\frac{4 c^2}{0.7} \approx 0.48883$ and the claim follows because
    $\beta \ge 0$.
\end{proof}

\begin{proof}[Proof of Claim~\ref{ineq:taylorh4}]
    From~\eqref{eq:taylorfourth}, we know that
    \begin{displaymath}
        h(1/4 - \beta/2) \le h(1/4) - \frac{\log_2{3}}{2} \beta
    \end{displaymath}
    Hence, we need to show that:
    \begin{displaymath}
        h(1/4) - \beta /2 - \alpha \beta \ge h(1/4) - \frac{\log_2{3}}{2} \beta
    \end{displaymath}
    Which means that:
    \begin{displaymath}
        0 \le \left(\frac{\log_2{3} - 1}{2} - \alpha\right) \beta \approx (0.292 - \alpha)\beta
    \end{displaymath}
    This holds when $\beta \ge 0$ and $\alpha \le 0.292$.
\end{proof}

\begin{proof}[Proof of Claim~\ref{ineq:taylorlast}]
    This is the moment, when the choice of $x$ is used. 
    First, let us rewrite the binomial coefficient as an entropy function.
    \begin{displaymath}
        \binom{3/4 - \alpha/2}{1/4 + c\alpha}_n \le 2^{n(3/4 - \alpha/2)h\left(
                \frac{1/4 - c\alpha}{3/4 - \alpha/2}
        \right)}.
    \end{displaymath}
    Hence, we need to prove
    \begin{displaymath}
        (3/4 - \alpha/2)h\left(
            \frac{1/4 + c\alpha}{3/4 - \alpha/2}
        \right)
        \le 
        \frac{3}{4} h(1/3) - 0.7 \cdot \alpha^2
        .
    \end{displaymath}
    Let us denote $\phi(\alpha) := \frac{1/4 + c\alpha}{3/4 - \alpha/2}$.
    Therefore, we need to show
    \begin{displaymath}
        (3/4 - \alpha/2)h\left(
                \phi(\alpha)
        \right)
        \le 
        \frac{3}{4} h(1/3) - 0.7 \cdot \alpha^2
        .
    \end{displaymath}
    The strategy behind the proof is straightforward. We bound $\phi(\alpha)$
    with Taylor expansion and then bound $h(\phi(\alpha))$. The inequality is
    technical, because we need to expand up to the $\Oh(\alpha^3)$ term.
    Let us use Taylor expansion of fraction inside binary entropy:
    \begin{displaymath}
        \phi(\alpha) := \frac{1/4+c\alpha}{3/4 - \alpha/2} = 1/3 + 
        \frac{2}{9}\left(6c+1  \right) \alpha +
        \frac{4}{27}\left(6c+1  \right) \alpha^2 +
        \frac{8}{81}\left(6c+1  \right) \alpha^3 +
        \Oh(\alpha^4)
        .
    \end{displaymath}
    Let $A := (12c+2)/9$ and $B:= (24c+4)/27$.
    Note, that $\frac{8}{81}\left(6c+1  \right) < 0.3$, therefore
    \begin{displaymath}
        \phi(\alpha) \le 1/3 + 
        A \alpha + B \alpha^2
        + \frac{0.3\alpha^3}{1-\alpha}
        .
    \end{displaymath}
    Because we assumed $\alpha < 0.1$ we can roughly bound:
    \begin{displaymath}
        \phi(\alpha) \le 1/3 + 
        A \alpha + B \alpha^2 +
        1/2 \cdot \alpha^3
        .
    \end{displaymath}
Then we plug in~\eqref{eq:taylorthird}, the Taylor expansion of $h(1/3 + x) \le h(1/3)
    + x + \kappa_1 \cdot x^2 + \kappa_2 \cdot x^3$ and have:
    \begin{displaymath}
        h(\phi(\alpha)) \le 
        h(1/3) +
        A \alpha + \left(B + \kappa_1 A^2
        \right)\alpha^2 + \frac{\kappa_2}{2} \alpha^3
        .
    \end{displaymath}
    Next we multiply it by $(\frac{3}{4} - \alpha/2)$ and have:
    \begin{align*}
        \left(\frac{3}{4} - \alpha/2\right) \cdot h\left(\phi(\alpha)\right)  \le &\; \frac{3}{4} h(1/3) + \\ 
        & \alpha  \left(\frac{3}{4} A - h(1/3)/2\right) + \\ 
        & \alpha^2  \left( \frac{3}{4} B - A/2 \right) + 
        \alpha^2 \left(\frac{3\kappa_1 \cdot A^2}{4}\right) +
        \frac{3 \kappa_2}{4} \alpha^3
        .
    \end{align*}
    The $x$ and the constant $c$ were chosen in such a way that $\frac{3}{4} A
    = h(1/3)/2$ and $\frac{3}{4} B = A/2$.  Hence:
    \begin{align*}
        \left(\frac{3}{4} - \alpha/2\right) \cdot h\left(\phi(\alpha)\right)  \le &\; \frac{3}{4}
        h(1/3) + \alpha^2 \left(\frac{3\kappa_1 \cdot A^2}{4}\right) +
        \alpha^3 \left(\frac{3 \kappa_2}{4} \right)
        .
    \end{align*}
    Moreover $3\kappa_1 \cdot A^2/4 < -0.9123$ and $3 \kappa_2/4 < 1.22$.
    Hence
    \begin{align*}
        \left(\frac{3}{4} - \alpha/2\right) \cdot h\left(\phi(\alpha)\right)  \le \frac{3}{4} h(1/3) -0.91 \alpha^2 + 1.22 \alpha^3
    \end{align*}
    Recall that we assumed that $\alpha < 0.1$, therefore:
    \begin{align*}
        \left(\frac{3}{4} - \alpha/2\right) \cdot h\left(\phi(\alpha)\right)  \le \frac{3}{4} h(1/3) -0.7 \alpha^2
    \end{align*}
    which we needed to prove.
\end{proof}

\section{Problems Definitions}\label{sec:problems}
\defproblem{$4$-SUM}
{Sets $A,B,C,D$ of integers and a target integer $t$}
{Find $a \in A$, $b \in B$, $c \in C$, $d \in D$ such that $a+b+c+d=t$.}

\defproblem{Binary Integer Programming (BIP)}
{Vectors $v,a^1,\ldots,a^d \in [m]^n$ and integers $u_1,\ldots,u_d \in [m]$}
{Find $x \in \mathbb{Z}^n$, such that
	\begin{equation*}
	\begin{array}{ll@{}ll}
	\text{minimize}  & \displaystyle \langle v, x\rangle 
	&\\
	\text{subject to}& \displaystyle \langle a^j,x\rangle \le u_j && \text{for all } j \in [d]\\
	&
	x_{i} \in \{0,1\} && \text{for all } i \in [n]
	.
	\end{array}
	\end{equation*}
}

\defproblem{Exact Node Weighted $P_4$}
{A node weighted, undirected graph $G$.}
{Decide if there exists a simple path on $4$ vertices with total weight equal exactly $0$.}


\defproblem{Knapsack}
{A set of $n$ items $\{(v_1,w_1),\ldots, (v_n,w_n)\}$ }
{Find $x \in \mathbb{Z}^n$ such that:
    \begin{equation*}
        \begin{array}{ll@{}ll}
            \text{maximize}  & \langle v, x \rangle
            &\\
            \text{subject to}& \langle w, x \rangle \leq t,  &\\
            & x_{i} \in \{0,1\}^n && \text{for all } i \in [n]
            .
        \end{array}
    \end{equation*}
}

\defproblem{Orthogonal Vectors (OV)}
{Two sets of vectors $\mathcal{A},\mathcal{B} \subseteq \{0,1\}^d$}
{Decide if there exists a pair $a \in \mathcal{A}$ and $b \in \mathcal{B}$ such
	that $\langle a,b \rangle = 0$.}

\defproblem{Subset Sum}
{A set of $n$ integers $\{w_1,\ldots,w_n\}$ and integer $t$ }
{Decide if there exists $x_1,\ldots,x_n \in \{0,1\}$, such that $\sum_{i=1}^n x_i w_i = t$.}

\end{appendices}

\end{document}